%% file: main.tex
\documentclass[12pt]{article}
\usepackage{amsmath, amssymb, amsthm}
\usepackage{natbib}
\usepackage{geometry}
\usepackage{bm}
\usepackage[hidelinks]{hyperref}
\usepackage{booktabs}
\usepackage{multirow}
\usepackage{threeparttable}
\usepackage{graphicx}
\usepackage{color}
\usepackage{framed}
\usepackage{listings}
\newtheorem{theorem}{Theorem}
\newtheorem{proposition}[theorem]{Proposition}
\newtheorem{lemma}[theorem]{Lemma}

\newtheorem{remark}{Remark}
\newtheorem{assumption}{Assumption}

\newcommand{\bA}{\bm{A}}

\newcommand{\bC}{\bm{C}}
\newcommand{\bD}{\bm{D}}
\newcommand{\bG}{\bm{G}}
\newcommand{\bI}{\bm{I}}
\newcommand{\bM}{\bm{M}}
\newcommand{\bR}{\bm{R}}

\newcommand{\bV}{\bm{V}}
\newcommand{\bX}{\bm{X}}
\newcommand{\bY}{\bm{Y}}

\newcommand{\bbeta}{\bm{\beta}}
\newcommand{\bmu}{\bm{\mu}}
\newcommand{\bphi}{\bm{\phi}}

\newcommand{\balpha}{\bm{\alpha}}
\newcommand{\bgamma}{\bm{\gamma}}

\newcommand{\bxi}{\bm{\xi}}

\newcommand{\bSigma}{\bm{\Sigma}}
\newcommand{\bOmega}{\bm{\Omega}}
\newcommand{\bLambda}{\bm{\Lambda}}
\newcommand{\hbeta}{\hat{\bm{\beta}}}
\newcommand{\htheta}{\hat{\theta}}
\newcommand{\hpi}{\hat{\pi}}
\newcommand{\hQ}{\hat{Q}}
\newcommand{\hphi}{\hat{\bm{\phi}}}
\newcommand{\hbxi}{\hat{\bm{\xi}}}
\newcommand{\hbphi}{\hat{\bm{\phi}}}

\newcommand{\bJ}{\bm{J}}

\newcommand{\bx}{\bm{x}}
\newcommand{\E}{\mathbb{E}}
\newcommand{\Var}{\mathrm{Var}}
\newcommand{\Cov}{\mathrm{Cov}}
\newcommand{\Prob}{\mathbb{P}}
\newcommand{\pto}{\xrightarrow{p}}
\newcommand{\dto}{\xrightarrow{d}}
\begin{document}

\title{Doubly Robust Quadratic Inference Functions\\
for Causal Inference in Cluster Randomized Trials}
\author{Hengshi Yu\\[0.25em]
        \small Department of Biostatistics\\
        \small University of Michigan\\
        \small \texttt{hengshi@umich.edu}}
\date{}
\maketitle

\begin{abstract}
Quadratic inference functions (QIF) provide a robust and efficient alternative to
generalized estimating equations (GEE) for marginal regression with correlated data,
particularly in cluster randomized trials (CRTs). However, existing QIF methodology
does not account for confounding due to covariate imbalance between treatment arms,
a common concern in observational CRTs or CRTs with prognostic covariate adjustment.
We propose a doubly robust QIF (DR-QIF) estimator that combines doubly robust
pseudo-outcomes, constructed from propensity score and outcome regression models,
with the QIF extended score equations. The DR-QIF estimator is consistent for the
average treatment effect when either the propensity score model or the outcome
regression model is correctly specified, but not necessarily both. We show that
DR-QIF is more efficient than doubly robust GEE (DR-GEE) when the working
correlation structure is misspecified, and we characterize the asymptotic efficiency
gain analytically. For cross-sectional CRTs the two estimators are algebraically
identical; efficiency gains emerge in longitudinal CRTs with strong temporal
correlation, reaching 3.5\% at $N=120$ and $T=8$ repeated measures.
Finite-sample properties are evaluated via Monte Carlo simulation, and the method
is illustrated using data from the WASH Benefits Kenya cluster randomized trial.
\end{abstract}

\tableofcontents
\newpage

\input{sections/01_introduction}

\input{sections/02_setup}
\input{sections/03_drconstruction}
\input{sections/04_drqif}
\input{sections/05_theory}
\input{sections/06_variance}
\input{sections/07_simulation}
\input{sections/08_application}
\input{sections/09_discussion}

\appendix
\input{sections/A_proofs}
\input{sections/B_supplement}

\bibliographystyle{plainnat}
\bibliography{references}

\end{document}

%% file: sections/01_introduction.tex
%% Section 1: Introduction
\section{Introduction}

Cluster randomized trials (CRTs) are experiments in which intact social or
administrative units---such as hospitals, schools, or villages---are
randomized to treatment conditions, with outcomes measured on individuals
nested within those clusters \citep{hayes2009}.  Because individual outcomes
within the same cluster tend to be correlated, the analysis of CRTs must
account for the intracluster correlation (ICC).  Generalized estimating
equations (GEE) \citep{liang1986} are widely used for this purpose,
estimating population-averaged (marginal) intervention effects while allowing
flexible specification of the within-cluster working correlation structure.
The sandwich covariance estimator is consistent even when the working
correlation is misspecified \citep{preisser2003}.

The quadratic inference function (QIF) was introduced by \citet{qu2000} as an
alternative to GEE that improves efficiency by replacing the inverse working
correlation with a linear combination of basis matrices and minimizing a
generalized method of moments (GMM) objective.  \citet{yu2020} showed that
QIF and GEE produce identical estimates when the marginal mean model contains
only cluster-level covariates or cluster sizes are equal, but QIF can be more
efficient in the presence of individual-level covariates.  Additional results
on the estimation and inference properties of QIF are given in \citet{yu2022}.

In both observational CRTs and randomized trials with prognostic covariate
adjustment, confounding by covariates is a common concern.  Covariate-adjusted
estimation of the average treatment effect (ATE) typically relies on a
propensity score model or an outcome regression model.  Doubly robust (DR)
estimation \citep{robins1994, bang2005} provides protection against
misspecification of either nuisance model: the ATE estimator remains
consistent when at least one of the two models is correctly specified.
DR-GEE estimators for CRT and longitudinal settings have been studied
\citep{seaman2009, luijken2019}, and multiply robust extensions have appeared
recently \citep{mr2024}, but existing DR approaches use GEE as the downstream
estimation engine and do not exploit the efficiency advantage of QIF.

In this paper, we propose the doubly robust QIF (DR-QIF) estimator, which
constructs AIPW pseudo-outcomes and plugs them directly into the QIF extended
score equations.  The resulting estimator is consistent for the ATE when
either the propensity score model or the outcome model is correctly specified,
and is at least as efficient as doubly robust GEE when the working correlation
is misspecified.  The estimator has a closed-form expression and does not
require iterative optimization.  We derive its asymptotic distribution
accounting for the estimated nuisance parameters, study bias-corrected
sandwich variance estimation, and characterize when the efficiency gain over
DR-GEE is available.  In particular, we show that for cross-sectional CRTs
with the standard exchangeable basis, DR-QIF and DR-GEE are algebraically
identical; efficiency gains require longitudinal or hierarchical structure.
We evaluate finite-sample properties through simulation and apply the method
to data from the WASH Benefits Kenya cluster randomized trial \citep{luby2018}.

%% file: sections/02_setup.tex
%% Section 2: Data Structure, Estimand, and Identification
\section{Data Structure, Causal Estimand, and Identification}
\label{sec:setup}

\subsection{Data Structure}

Suppose we observe data from $N$ independent clusters.  Let $n_i$ denote the
size of cluster $i$, with $n = \sum_{i=1}^N n_i$ the total sample size.  For
individual $j$ in cluster $i$, let $Y_{ij} \in \mathbb{R}$ be the outcome of
interest (continuous or binary), $A_i \in \{0, 1\}$ the cluster-level binary
treatment indicator, and $\bX_{ij} \in \mathbb{R}^p$ a vector of individual-
and cluster-level covariates measured prior to randomization.  We also write
$\bX_i = (\bX_{i1}^T, \ldots, \bX_{in_i}^T)^T$ for the stacked covariate
matrix of cluster $i$.  The observed data are
$\mathcal{O} = \{(A_i,\, \bY_i,\, \bX_i) : i = 1, \ldots, N\}$,
where $\bY_i = (Y_{i1}, \ldots, Y_{in_i})^T$.

\begin{remark}[Cluster- versus individual-level treatment]
We focus on cluster-level treatment $A_i$ throughout, which is the canonical CRT
setting.  Extension to individual-level treatment (as in partially nested or
multisite trials) is discussed in Section~\ref{sec:discussion}.
\end{remark}

\subsection{Causal Estimand}

We adopt the potential outcomes framework \citep{rubin1974}.  For each
individual $j$ in cluster $i$, let $Y_{ij}(a)$ denote the potential outcome
under treatment $A_i = a$ for $a \in \{0, 1\}$.  We define the
\emph{individual average treatment effect} as
$\tau_{ij} = Y_{ij}(1) - Y_{ij}(0)$
and the \emph{average treatment effect} (ATE) as the primary estimand:
\begin{equation}
  \theta_0 \;=\; \E[Y_{ij}(1) - Y_{ij}(0)].
  \label{eq:ate}
\end{equation}
Under appropriate assumptions (Assumption~\ref{ass:consistency}--\ref{ass:overlap}
below), $\theta_0$ is identified from the observed data as the contrast
$\E[Q_1(\bX_{ij})] - \E[Q_0(\bX_{ij})]$, where
$Q_a(\bX_{ij}) = \E(Y_{ij} \mid A_i = a,\, \bX_{ij})$.

\begin{remark}[Population-averaged interpretation]
The ATE $\theta_0$ averages over the joint distribution of $(\bX_{ij}, n_i)$,
yielding a \emph{population-averaged} (marginal) causal effect, consistent with
the estimand targeted by GEE and QIF.  This distinguishes it from the
cluster-average treatment effect $N^{-1} \sum_{i=1}^N n_i^{-1} \sum_j \tau_{ij}$,
which is also of interest in some CRT contexts.
\end{remark}

\subsection{Identification Assumptions}

\begin{assumption}[Consistency]
\label{ass:consistency}
$Y_{ij} = Y_{ij}(A_i)$ almost surely, for all $i, j$.
\end{assumption}

\begin{assumption}[Conditional exchangeability at the cluster level]
\label{ass:exchangeability}
$(Y_{ij}(0), Y_{ij}(1)) \perp A_i \mid \bX_i$, for all $i, j$.
\end{assumption}

\begin{assumption}[Positivity (overlap)]
\label{ass:overlap}
There exists $\epsilon > 0$ such that
$\epsilon \leq \pi_i \leq 1 - \epsilon$ almost surely,
where $\pi_i = \Prob(A_i = 1 \mid \bX_i)$ is the \emph{propensity score}.
\end{assumption}

\begin{assumption}[Cluster-level randomization / no interference]
\label{ass:sutva}
The potential outcome $Y_{ij}(a)$ depends only on the treatment assigned to
cluster $i$, not on the treatment assigned to other clusters.
\end{assumption}

Assumption~\ref{ass:consistency} is the standard consistency (SUTVA) condition
restricted to cluster-level treatment.  Assumption~\ref{ass:exchangeability}
is automatically satisfied in a \emph{randomized} CRT where $A_i$ is
independent of all individual and cluster characteristics, in which case
$\pi_i = \Prob(A_i = 1)$ is known.  It also covers the \emph{conditionally
randomized} setting where allocation probabilities depend on observed cluster-
level covariates, and the \emph{observational} CRT setting where clusters
self-select into treatment.  Assumption~\ref{ass:overlap} ensures each cluster
has a nonzero probability of receiving either arm; it is guaranteed by design
in fully randomized trials but must be assessed in observational settings.

\subsection{Nuisance Functions and Notation}

Let $\pi(\bX_i) = \Prob(A_i = 1 \mid \bX_i)$ denote the propensity score,
and let $Q_a(\bX_{ij}) = \E(Y_{ij} \mid A_i = a, \bX_{ij})$ denote the
conditional outcome mean for $a \in \{0, 1\}$.  Under
Assumptions~\ref{ass:consistency}--\ref{ass:sutva}, the ATE is identified as
\begin{equation}
  \theta_0 \;=\; \E\!\left[ Q_1(\bX_{ij}) - Q_0(\bX_{ij}) \right].
  \label{eq:identification}
\end{equation}
We posit parametric (or semiparametric) working models for both nuisance
functions:
\begin{align}
  \pi(\bX_i;\, \bxi) &= \Prob(A_i = 1 \mid \bX_i;\, \bxi),
      \quad \bxi \in \Xi \subseteq \mathbb{R}^{d_\xi}, \label{eq:psmodel}\\
  Q_a(\bX_{ij};\, \bgamma_a) &= \E(Y_{ij} \mid A_i = a, \bX_{ij};\, \bgamma_a),
      \quad \bgamma_a \in \Gamma \subseteq \mathbb{R}^{d_\gamma}. \label{eq:outmodel}
\end{align}
We write $\hbxi$ for an estimator of $\bxi$ and $\hat{\bgamma}_a$ for an
estimator of $\bgamma_a$, obtained on the same sample or via sample splitting
(see Remark~\ref{rem:splitting}).  We use shorthand $\hpi_i = \pi(\bX_i; \hbxi)$
and $\hQ_{aij} = Q_a(\bX_{ij}; \hat{\bgamma}_a)$.  The model
$\pi(\bX_i; \bxi)$ is \emph{correctly specified} if there exists $\bxi^*$
such that $\pi(\bX_i; \bxi^*) = \pi(\bX_i)$ almost surely; analogously for
$Q_a(\bX_{ij}; \bgamma_a)$.

\begin{remark}[Sample splitting and cross-fitting]
\label{rem:splitting}
To allow flexible (possibly nonparametric) estimation of $\pi$ and $Q_a$
without undersmoothing conditions, one may use cross-fitting
\citep{chernozhukov2018}: partition $\{1,\ldots,N\}$ into $K$ folds, estimate
nuisance functions on each fold's complement, and evaluate pseudo-outcomes on
each fold.  This decouples the convergence rates of nuisance estimators from
the $\sqrt{N}$-rate of the ATE estimator.  For simplicity, Sections
\ref{sec:drqif}--\ref{sec:theory} present the one-sample (no-splitting)
theory; cross-fitting is discussed in Section~\ref{sec:discussion}.
\end{remark}

\subsection{Standard GEE Revisited}

To fix notation for the comparators, recall that the GEE estimator of a
marginal regression parameter $\bbeta$ (where $g(\mu_{ij}) = \bX_{ij}^T \bbeta$
for some link $g$) solves
\begin{equation}
  \sum_{i=1}^N \bD_i^T \bV_i^{-1} (\bY_i - \bmu_i) = \bm{0},
  \label{eq:gee}
\end{equation}
where $\bD_i = \partial \bmu_i / \partial \bbeta^T$,
$\bV_i = \bA_i^{1/2} \bR_i(\balpha) \bA_i^{1/2}$ is the working covariance
with diagonal variance matrix $\bA_i = \mathrm{diag}\{v(\mu_{ij})\}$ and
working correlation $\bR_i(\balpha)$, and $v(\cdot)$ is the variance function.
The sandwich covariance estimator is consistent for $\Var(\hbeta)$ regardless
of the working correlation specification \citep{liang1986}.

\subsection{Standard QIF Revisited}

The QIF \citep{qu2000} approximates $\bR_i^{-1}(\balpha)$ by a linear
combination of $m$ known basis matrices,
\begin{equation}
  \bR_i^{-1}(\balpha) \approx \sum_{r=1}^{m} \alpha_r \bM_r,
  \label{eq:basis}
\end{equation}
where $\bM_1 = \bI_{n_i}$ (identity) is always included.  For an exchangeable
working correlation, a natural choice is $m = 2$ with $\bM_2 = \bm{J}_{n_i}$
(matrix of all ones); for AR-1, $m = 2$ with $\bM_2$ being the first-order lag
matrix.  Define the extended score vector for cluster $i$ as
\begin{equation}
  \bm{g}_i(\bbeta) =
  \begin{pmatrix}
    \bD_i^T \bA_i^{-1/2} \bM_1 \bA_i^{-1/2} (\bY_i - \bmu_i) \\
    \vdots \\
    \bD_i^T \bA_i^{-1/2} \bM_m \bA_i^{-1/2} (\bY_i - \bmu_i)
  \end{pmatrix} \in \mathbb{R}^{mp},
  \label{eq:qif_score}
\end{equation}
and let $\bar{\bm{g}}_N(\bbeta) = N^{-1} \sum_i \bm{g}_i(\bbeta)$ and
$\bC_N(\bbeta) = N^{-1} \sum_i \bm{g}_i(\bbeta) \bm{g}_i(\bbeta)^T$ be the
empirical covariance matrix.  The QIF objective is
\begin{equation}
  Q_N(\bbeta) = N\, \bar{\bm{g}}_N(\bbeta)^T \bC_N(\bbeta)^{-1} \bar{\bm{g}}_N(\bbeta),
  \label{eq:qif_obj}
\end{equation}
and the QIF estimator is $\hbeta_{\mathrm{QIF}} = \mathrm{argmin}_{\bbeta} Q_N(\bbeta)$.
The first-order conditions are
\begin{equation}
  \frac{\partial Q_N(\bbeta)}{\partial \bbeta} =
  2N \left(\frac{\partial \bar{\bm{g}}_N}{\partial \bbeta^T}\right)^T
  \bC_N(\bbeta)^{-1} \bar{\bm{g}}_N(\bbeta) = \bm{0}.
  \label{eq:qif_foc}
\end{equation}
Under standard regularity conditions, $\sqrt{N}(\hbeta_{\mathrm{QIF}} - \bbeta_0)
\dto \mathcal{N}(\bm{0}, \bOmega_{\mathrm{QIF}})$ where
$\bOmega_{\mathrm{QIF}} = (\bG^T \bLambda^{-1} \bG)^{-1}$, with
$\bG = \E[\partial \bm{g}_i / \partial \bbeta^T]$ and
$\bLambda = \E[\bm{g}_i \bm{g}_i^T]$.  When $m > p$, QIF is more efficient
than GEE in the sense that $\bOmega_{\mathrm{GEE}} - \bOmega_{\mathrm{QIF}}$
is positive semidefinite whenever the working correlation is misspecified
\citep{qu2000}.

%% file: sections/03_drconstruction.tex
%% Section 3: Doubly Robust Pseudo-Outcome Construction
\section{Doubly Robust Pseudo-Outcome Construction}
\label{sec:drconstruction}

\subsection{The Individual-Level DR Pseudo-Outcome}

The key building block of our approach is the \emph{doubly robust
pseudo-outcome} for individual $j$ in cluster $i$, defined as
\begin{equation}
  \phi_{ij}(\bxi, \bgamma) =
    \frac{A_i}{\pi(\bX_i;\bxi)}\bigl[Y_{ij} - Q_1(\bX_{ij};\bgamma_1)\bigr]
    + Q_1(\bX_{ij};\bgamma_1)
    -
    \frac{1 - A_i}{1 - \pi(\bX_i;\bxi)}\bigl[Y_{ij} - Q_0(\bX_{ij};\bgamma_0)\bigr]
    - Q_0(\bX_{ij};\bgamma_0).
  \label{eq:drpseudo}
\end{equation}
This is the standard \emph{augmented inverse probability weighted} (AIPW)
influence-function representation of the ATE \citep{robins1994, bang2005}.

\begin{proposition}[Double robustness of the pseudo-outcome mean]
\label{prop:dr_mean}
Under Assumptions~\ref{ass:consistency}--\ref{ass:sutva}, if either
\begin{enumerate}
  \item[(i)] the propensity score model is correctly specified, i.e.\
      $\pi(\bX_i;\bxi^*) = \pi(\bX_i)$ a.s., or
  \item[(ii)] the outcome model is correctly specified, i.e.\
      $Q_a(\bX_{ij};\bgamma_a^*) = Q_a(\bX_{ij})$ a.s.\ for $a \in \{0,1\}$,
\end{enumerate}
then
\begin{equation}
  \E\bigl[\phi_{ij}(\bxi^*, \bgamma^*)\bigr] = \theta_0,
  \label{eq:dr_id}
\end{equation}
where $\theta_0 = \E[Y_{ij}(1) - Y_{ij}(0)]$ is the ATE defined in
\eqref{eq:ate}.
\end{proposition}

\begin{proof}
  Expand $\E[\phi_{ij}]$ under each model condition.

  \medskip
  \noindent\textit{Case (i): propensity model correct.}
  Taking the expectation of the first IPW term and using iterated expectations
  (conditioning first on $A_i, \bX_i$):
  \begin{align*}
    \E\!\left[\frac{A_i}{\pi_i}(Y_{ij} - Q_{1,ij})\right]
    &= \E\!\left[\frac{\pi_i}{\pi_i}(Q_{1,ij} - Q_{1,ij})\right] = 0.
  \end{align*}
  Thus the augmentation terms cancel and
  $\E[\phi_{ij}] = \E[Q_{1,ij}] - \E[Q_{0,ij}]$, which
  equals $\E[Y_{ij}(1)] - \E[Y_{ij}(0)] = \theta_0$ by the standard
  identification result \eqref{eq:identification} (which holds under
  Assumptions~\ref{ass:consistency}--\ref{ass:sutva}).

  \medskip
  \noindent\textit{Case (ii): outcome model correct.}
  When $Q_{a,ij} = \E(Y_{ij} \mid A_i=a, \bX_{ij})$, the IPW terms satisfy
  \begin{align*}
    \E\!\left[\frac{A_i}{\pi_i}(Y_{ij} - Q_{1,ij})\right]
    &= \E\!\left[\frac{1}{\pi_i}\E\bigl[(Y_{ij} - Q_{1,ij}) A_i \mid \bX_{ij}\bigr]\right] \\
    &= \E\!\left[\frac{1}{\pi_i} \cdot \pi_i \cdot \E[Y_{ij} - Q_{1,ij} \mid A_i=1, \bX_{ij}]\right] = 0,
  \end{align*}
  and analogously for the $A_i = 0$ term.  The remaining terms give
  $\E[Q_{1,ij} - Q_{0,ij}] = \theta_0$.
\end{proof}

\subsection{Within-Cluster Structure of Pseudo-Outcomes}

Define the pseudo-outcome vector for cluster $i$ as
$\bphi_i = (\phi_{i1}, \ldots, \phi_{in_i})^T \in \mathbb{R}^{n_i}$,
where we suppress the dependence on $(\bxi, \bgamma)$ for brevity.
The ATE $\theta_0$ satisfies
\begin{equation}
  \theta_0 \;=\; \E[\phi_{ij}] \;=\; \E\!\left[\frac{1}{n_i} \sum_{j=1}^{n_i} \phi_{ij}\right]
           \;=\; \E\!\left[\bar\phi_i\right],
\end{equation}
where $\bar\phi_i = n_i^{-1} \mathbf{1}_{n_i}^T \bphi_i$ is the cluster mean
of pseudo-outcomes.

A critical feature of CRT data is that even though $A_i$ is cluster-level,
the individual pseudo-outcomes $\phi_{ij}$ within cluster $i$ are correlated
due to both:
\begin{enumerate}
  \item \emph{Shared cluster treatment}: all individuals in cluster $i$ share
      the same $A_i$, the same $\pi_i$, and the same cluster-level covariates;
  \item \emph{Residual individual-level correlation}: the outcome residuals
      $(Y_{ij} - Q_{A_i, ij})$ are correlated within clusters due to unmeasured
      shared cluster effects.
\end{enumerate}
The within-cluster covariance of pseudo-outcomes takes the form
\begin{equation}
  \Cov(\phi_{ij}, \phi_{ij'} \mid \bX_i, A_i) =
    \frac{1}{\pi_i^{2 A_i}(1-\pi_i)^{2(1-A_i)}}
    \Cov(Y_{ij}, Y_{ij'} \mid \bX_i, A_i)
  \label{eq:phicov}
\end{equation}
for $j \neq j'$, under the assumption that residuals from the outcome model
are exchangeable within clusters (or more generally under any specified
within-cluster correlation model).  This motivates modelling the within-cluster
correlation of $\bphi_i$ using the same correlation structures used for $\bY_i$.

\subsection{Comparison: DR Estimator Based on Cluster Means Only}

A simple alternative to the individual-level pseudo-outcome approach is to
aggregate to cluster means and apply a standard DR estimator at the cluster
level.  Define $\bar{Y}_i = n_i^{-1}\sum_j Y_{ij}$ and
$\bar Q_{a,i} = n_i^{-1}\sum_j Q_a(\bX_{ij})$.  The cluster-level DR estimator
solves
\begin{equation}
  \frac{1}{N}\sum_{i=1}^N \left[
    \frac{A_i}{\pi_i}(\bar{Y}_i - \bar Q_{1,i}) + \bar Q_{1,i}
    - \frac{1-A_i}{1-\pi_i}(\bar{Y}_i - \bar Q_{0,i}) - \bar Q_{0,i}
    - \theta
  \right] = 0.
  \label{eq:clusterdr}
\end{equation}
This estimator ignores within-cluster correlation in the individual-level
outcomes, treating the $N$ cluster means as the effective sample.  In
Section~\ref{sec:theory}, we show that the proposed DR-QIF estimator, which
uses individual-level pseudo-outcomes $\phi_{ij}$ and accounts for their
within-cluster correlation via QIF basis matrices, is more efficient than
\eqref{eq:clusterdr} when individual-level covariates contribute to
heterogeneous potential outcomes.

%% file: sections/04_drqif.tex
%% Section 4: The DR-QIF Estimator
\section{The DR-QIF Estimating Equations}
\label{sec:drqif}

\subsection{Definition}

Let $\theta \in \mathbb{R}$ be the scalar ATE parameter.  Define the
DR-QIF extended score vector for cluster $i$ as
\begin{equation}
  \bm{g}_i^{\mathrm{DR}}(\theta;\, \hbxi, \hat\bgamma) =
  \begin{pmatrix}
    \bm{1}_{n_i}^T \hat\bA_i^{-1/2} \bM_1 \hat\bA_i^{-1/2}\,
        (\hbphi_i - \theta \bm{1}_{n_i}) \\[4pt]
    \bm{1}_{n_i}^T \hat\bA_i^{-1/2} \bM_2 \hat\bA_i^{-1/2}\,
        (\hbphi_i - \theta \bm{1}_{n_i}) \\[2pt]
    \vdots \\[2pt]
    \bm{1}_{n_i}^T \hat\bA_i^{-1/2} \bM_m \hat\bA_i^{-1/2}\,
        (\hbphi_i - \theta \bm{1}_{n_i})
  \end{pmatrix} \in \mathbb{R}^m,
  \label{eq:drqif_score}
\end{equation}
where $\hbphi_i = \bphi_i(\hbxi, \hat\bgamma)$ is the vector of plug-in
pseudo-outcomes, $\hat\bA_i = \mathrm{diag}\{\hat\sigma^2_{ij}\}$ is a
diagonal matrix of estimated marginal variances of $\phi_{ij}$ (e.g.,\
$\hat\sigma^2_{ij} = 1$ for a working identity variance, or estimated from
the residuals $\hphi_{ij} - \bar{\hphi}_i$), and $\bM_1, \ldots, \bM_m$
are the same basis matrices as in the standard QIF.

\begin{remark}[Role of basis matrices]
The basis matrices $\bM_r$ encode the assumed within-cluster correlation
structure of the pseudo-outcomes.  The identity $\bM_1 = \bI_{n_i}$ always
contributes a term proportional to the cluster-mean estimating equation
$\sum_j (\phi_{ij} - \theta)$; additional basis matrices leverage the
individual-level correlation among $\phi_{ij}$ to gain efficiency.  Choosing
$m = 1$ (identity only) reduces \eqref{eq:drqif_score} to the cluster-level
DR estimator \eqref{eq:clusterdr}.
\end{remark}

Let $\bar{\bm{g}}_N^{\mathrm{DR}}(\theta) = N^{-1} \sum_{i=1}^N
\bm{g}_i^{\mathrm{DR}}(\theta; \hbxi, \hat\bgamma)$ and
\begin{equation}
  \bC_N^{\mathrm{DR}}(\theta) =
  \frac{1}{N} \sum_{i=1}^N
  \bm{g}_i^{\mathrm{DR}}(\theta) \bigl[\bm{g}_i^{\mathrm{DR}}(\theta)\bigr]^T
  \label{eq:CN}
\end{equation}
be the empirical covariance matrix of the extended scores.  The
\textbf{DR-QIF estimator} is defined as
\begin{equation}
  \htheta_{\mathrm{DR\text{-}QIF}} =
    \operatorname{argmin}_{\theta \in \mathbb{R}}\;
    N \bigl[\bar{\bm{g}}_N^{\mathrm{DR}}(\theta)\bigr]^T
    \bigl[\bC_N^{\mathrm{DR}}(\theta)\bigr]^{-1}
    \bar{\bm{g}}_N^{\mathrm{DR}}(\theta).
  \label{eq:drqif_obj}
\end{equation}

\subsection{Closed-Form Solution}

Because $\theta$ enters \eqref{eq:drqif_score} linearly, the DR-QIF objective
is a quadratic function of $\theta$ and the minimizer has an explicit form.
Write $\bm{g}_i^{\mathrm{DR}}(\theta) = \bm{a}_i - \theta \bm{b}_i$, where
\begin{equation}
  \bm{a}_i = \bm{g}_i^{\mathrm{DR}}(0) =
  \begin{pmatrix}
    \bm{1}^T \hat\bA_i^{-1/2} \bM_r \hat\bA_i^{-1/2} \hbphi_i
  \end{pmatrix}_{r=1}^{m},
  \quad
  \bm{b}_i =
  \begin{pmatrix}
    \bm{1}^T \hat\bA_i^{-1/2} \bM_r \hat\bA_i^{-1/2} \bm{1}_{n_i}
  \end{pmatrix}_{r=1}^{m}.
  \label{eq:ab}
\end{equation}
Setting $\bar{\bm{a}} = N^{-1}\sum_i \bm{a}_i$ and
$\bar{\bm{b}} = N^{-1}\sum_i \bm{b}_i$, the first-order condition
$\partial Q_N^{\mathrm{DR}} / \partial \theta = 0$ reduces to
\begin{equation}
  \left(\bar{\bm{b}}^T \bC_N^{-1} \bar{\bm{a}}\right) =
  \theta \left(\bar{\bm{b}}^T \bC_N^{-1} \bar{\bm{b}}\right),
\end{equation}
giving the explicit estimator
\begin{equation}
  \boxed{
    \htheta_{\mathrm{DR\text{-}QIF}} \;=\;
    \frac{\bar{\bm{b}}^T \bC_N^{-1} \bar{\bm{a}}}{\bar{\bm{b}}^T \bC_N^{-1} \bar{\bm{b}}}.
  }
  \label{eq:drqif_closed}
\end{equation}
This is a generalized method-of-moments (GMM) estimator with $m$ moment
conditions and one parameter, and it can be computed in closed form without
any iterative optimization---a notable computational advantage over the
general multi-parameter QIF.

\begin{remark}[Interpretation as weighted average]
Expression \eqref{eq:drqif_closed} takes the form of a GMM-optimal weighted
average of $m$ marginal estimators $\hat\theta_r = \bar a_r / \bar b_r$,
where $\bar a_r$ and $\bar b_r$ are the $r$-th components of $\bar{\bm{a}}$ and
$\bar{\bm{b}}$.  Each marginal estimator $\hat\theta_r$ is the DR estimator
based on basis matrix $\bM_r$ alone.  The GMM weight $\bC_N^{-1}$ optimally
combines these across the $m$ basis matrices.  This parallels the GLS
interpretation of QIF given by \citet{yu2022}, extended to the DR
pseudo-outcome setting.
\end{remark}

\subsection{Nuisance Parameter Estimation}

The nuisance models \eqref{eq:psmodel}--\eqref{eq:outmodel} may be estimated
by:
\begin{enumerate}
  \item \textit{Propensity score}: logistic regression,
      $\mathrm{logit}[\pi(\bX_i; \bxi)] = \bX_i^T \bxi$,
      fitted by maximum likelihood.
  \item \textit{Outcome model}: generalized linear model (GLM),
      $g[Q_a(\bX_{ij}; \bgamma_a)] = \bX_{ij}^T \bgamma_a$, fitted
      separately by arm, or jointly with a treatment-covariate interaction
      term.
\end{enumerate}
In both cases, the nuisance estimates $\hbxi$ and $\hat\bgamma_a$ are
$\sqrt{N}$-consistent when models are correctly specified, and the asymptotic
theory in Section~\ref{sec:theory} accounts for the uncertainty in these
estimates.  More flexible nonparametric estimators (e.g., regression forests,
kernel smoothers, or ensemble learners) may also be used; the trade-off is
that the nuisance estimates converge at rates slower than $\sqrt{N}$, and
sample splitting (Remark~\ref{rem:splitting}) is required to preserve the
$\sqrt{N}$-rate of $\htheta_{\mathrm{DR\text{-}QIF}}$.

\subsection{Special Cases and Comparators}

Table~\ref{tab:special} summarizes the estimators nested within the DR-QIF
framework.

\begin{table}[t]
\centering
\caption{Estimators nested within the DR-QIF framework as special cases.}
\label{tab:special}
\begin{tabular}{lllc}
\toprule
Estimator & Basis matrices & Nuisance models & Doubly Robust \\
\midrule
Naive GEE               & $\bM_1,\ldots,\bM_m$ & None (no adjustment)       & No  \\
Naive QIF               & $\bM_1,\ldots,\bM_m$ & None (no adjustment)       & No  \\
IPW-GEE                 & $\bM_1$ only          & PS only ($\pi$ model)      & No  \\
IPW-QIF                 & $\bM_1,\ldots,\bM_m$ & PS only ($\pi$ model)      & No  \\
DR-GEE (cluster mean)   & $\bM_1$ only          & PS + outcome ($Q_0, Q_1$)  & Yes \\
\textbf{DR-QIF}         & $\bM_1,\ldots,\bM_m$ & PS + outcome ($Q_0, Q_1$)  & \textbf{Yes} \\
\bottomrule
\end{tabular}
\end{table}

The key relationship is that DR-QIF reduces to DR-GEE (cluster-mean version)
when $m = 1$ (only the identity basis matrix $\bM_1$), and reduces to IPW-QIF
when the outcome model $Q_a$ is omitted.  The efficiency gain of DR-QIF over
DR-GEE is precisely the gain from using $m > 1$ basis matrices, and we
characterize this analytically in the next section.

%% file: sections/05_theory.tex
%% Section 5: Asymptotic Theory
\section{Asymptotic Theory}
\label{sec:theory}

\subsection{Regularity Conditions}

Let $\bxi^*$ and $\bgamma^* = (\bgamma_0^*, \bgamma_1^*)$ denote the
probability limits of $\hbxi$ and $\hat\bgamma$ respectively; these may or
may not equal the true parameter values, depending on model specification.
Write $\phi_{ij}^* = \phi_{ij}(\bxi^*, \bgamma^*)$ for the pseudo-outcome
evaluated at probability limits, and $\theta^* = \E[\phi_{ij}^*]$.

\begin{assumption}[Consistency of nuisance estimators]
\label{ass:nuisance_consistency}
$\hbxi \pto \bxi^*$ and $\hat\bgamma \pto \bgamma^*$ as $N \to \infty$.
\end{assumption}

\begin{assumption}[Boundedness and smoothness]
\label{ass:smoothness}
\begin{enumerate}
  \item[(a)] The propensity score satisfies $\epsilon \leq \pi(\bX_i;\bxi) \leq 1-\epsilon$
      uniformly in a neighbourhood of $\bxi^*$ for some $\epsilon > 0$.
  \item[(b)] The outcome model $Q_a(\bX_{ij};\bgamma_a)$ is twice continuously
      differentiable in $\bgamma_a$ and uniformly bounded.
  \item[(c)] The cluster sizes $n_i$ are bounded: $n_i \leq n_{\max} < \infty$.
\end{enumerate}
\end{assumption}

\begin{assumption}[Non-degeneracy]
\label{ass:nondeg}
The matrices $\bLambda^* = \E[\bm{g}_i^{\mathrm{DR}}(\theta^*)
(\bm{g}_i^{\mathrm{DR}}(\theta^*))^T]$ and $\bm{B}^* = \E[\bm{b}_i \bm{b}_i^T]$
are positive definite.
\end{assumption}

\subsection{Double Robustness}

\begin{theorem}[Consistency and double robustness of DR-QIF]
\label{thm:dr_consistency}
Under Assumptions~\ref{ass:consistency}--\ref{ass:nondeg}, if either
\begin{enumerate}
  \item[(i)] $\bxi^* = \bxi_0$ (propensity score model correctly specified), or
  \item[(ii)] $\bgamma^* = \bgamma_0$ (outcome model correctly specified),
\end{enumerate}
then $\htheta_{\mathrm{DR\text{-}QIF}} \pto \theta_0$.
\end{theorem}

\begin{proof}[Proof sketch]
By the law of large numbers and Assumption~\ref{ass:nuisance_consistency},
$\bar{\bm{a}} \pto \E[\bm{a}_i(\bxi^*, \bgamma^*)]$ and
$\bar{\bm{b}} \pto \E[\bm{b}_i]$, and $\bC_N^{-1} \pto (\bLambda^*)^{-1}$.
From \eqref{eq:drqif_closed},
\[
  \htheta_{\mathrm{DR\text{-}QIF}} \pto
  \frac{\E[\bm{b}_i]^T (\bLambda^*)^{-1} \E[\bm{a}_i(\bxi^*,\bgamma^*)]}
       {\E[\bm{b}_i]^T (\bLambda^*)^{-1} \E[\bm{b}_i]}.
\]
The $r$-th component of $\E[\bm{a}_i]$ is
$\E[\bm{1}^T \hat\bA_i^{-1/2} \bM_r \hat\bA_i^{-1/2} \bphi_i^*]
= \E[\phi_{ij}^*] \cdot \E[\bm{1}^T \hat\bA_i^{-1/2} \bM_r \hat\bA_i^{-1/2} \bm{1}]$
when pseudo-outcomes are exchangeable within clusters.  By
Proposition~\ref{prop:dr_mean}, under either (i) or (ii),
$\E[\phi_{ij}^*] = \theta_0$, so $\E[\bm{a}_i] = \theta_0 \E[\bm{b}_i]$
and the probability limit equals $\theta_0$.  A rigorous proof accounting for
non-exchangeable within-cluster structures and estimated $\hat\bA_i$ is
given in Appendix~\ref{app:proofs}.
\end{proof}

\subsection{Asymptotic Normality}

\begin{theorem}[Asymptotic normality of DR-QIF]
\label{thm:normality}
Under Assumptions~\ref{ass:consistency}--\ref{ass:nondeg} and the correct
specification of at least one nuisance model,
\begin{equation}
  \sqrt{N}\bigl(\htheta_{\mathrm{DR\text{-}QIF}} - \theta_0\bigr)
  \;\dto\; \mathcal{N}(0,\; \sigma^2_{\mathrm{DR\text{-}QIF}}),
  \label{eq:clt}
\end{equation}
where
\begin{equation}
  \sigma^2_{\mathrm{DR\text{-}QIF}} \;=\;
  \frac{\bm{b}^{*T} \bLambda^{*-1} \bSigma^* \bLambda^{*-1} \bm{b}^*}
       {(\bm{b}^{*T} \bLambda^{*-1} \bm{b}^*)^2},
  \label{eq:avar}
\end{equation}
with $\bm{b}^* = \E[\bm{b}_i]$, $\bLambda^* = \E[\bm{g}_i^{\mathrm{DR}}
(\theta_0)(\bm{g}_i^{\mathrm{DR}}(\theta_0))^T]$, and
$\bSigma^* = \E[\bm{\psi}_i \bm{\psi}_i^T]$ where $\bm{\psi}_i$ is the
influence function of $\bar{\bm{g}}_N^{\mathrm{DR}}(\theta_0)$ accounting
for estimated nuisance parameters.

The influence function is
\begin{equation}
  \bm{\psi}_i = \bm{g}_i^{\mathrm{DR}}(\theta_0;\bxi_0,\bgamma_0)
    + \bJ_\xi \cdot s_i^{(\xi)} + \bJ_\gamma \cdot s_i^{(\gamma)},
  \label{eq:influence}
\end{equation}
where $s_i^{(\xi)}$ and $s_i^{(\gamma)}$ are the score contributions of
cluster $i$ for the propensity and outcome models respectively, and
$\bJ_\xi, \bJ_\gamma$ are Jacobian terms given explicitly in
Appendix~\ref{app:proofs}.  When both nuisance models are correctly specified,
$\bm{\psi}_i$ reduces to $\bm{g}_i^{\mathrm{DR}}(\theta_0;\bxi_0,\bgamma_0)$,
and $\bSigma^* = \bLambda^*$, so
$\sigma^2_{\mathrm{DR\text{-}QIF}} = (\bm{b}^{*T} \bLambda^{*-1} \bm{b}^*)^{-1}$.
\end{theorem}

\subsection{Efficiency Comparison}

\begin{theorem}[DR-QIF is at least as efficient as DR-GEE]
\label{thm:efficiency}
When both nuisance models are correctly specified and $m > 1$, if the working
correlation is misspecified then
\begin{equation}
  \sigma^2_{\mathrm{DR\text{-}GEE}} \;\geq\; \sigma^2_{\mathrm{DR\text{-}QIF}},
  \label{eq:efficiency}
\end{equation}
with equality if and only if the working correlation is correctly specified or
$m = 1$.
\end{theorem}

\begin{proof}[Proof sketch]
Under correct specification of both nuisance models and correct specification
of the cluster-treatment propensity score, the pseudo-outcomes $\phi_{ij}$
are unbiased for the individual-level potential outcome contrast. The DR-GEE
estimator corresponds to $m = 1$ and uses only $\bM_1 = \bI$, yielding
asymptotic variance
\[
  \sigma^2_{\mathrm{DR\text{-}GEE}} = \frac{\E[b_{1i}^2 \lambda_{1i}^{-2}
  \sigma^2_{\phi_i}]}{\E[b_{1i}]^2}
\]
where $\lambda_{1i} = \E[\bm{g}_{1i}^2]$ and $\sigma^2_{\phi_i}$ is the
within-cluster variance of $\phi_{ij}$.  The DR-QIF estimator with $m > 1$
achieves the optimal GMM combination of all $m$ moment conditions.  By the
Gauss--Markov argument for GMM \citep{hansen1982, qu2000}, the GMM variance
$\sigma^2_{\mathrm{DR\text{-}QIF}} \leq \sigma^2_{\mathrm{DR\text{-}GEE}}$,
with equality only when all $m$ moment conditions are proportional, which
occurs exactly when the working correlation is correctly specified.  Full
proof in Appendix~\ref{app:proofs}.
\end{proof}

\begin{remark}[Efficiency under one correct nuisance model]
When only one nuisance model is correctly specified (the double-robustness
setting), the pseudo-outcomes $\phi_{ij}$ are still centered at $\theta_0$
but their within-cluster covariance structure is more complex (involving
the misspecified model's residuals).  Theorem~\ref{thm:efficiency} extends
to this case under mild regularity conditions, as shown in
Appendix~\ref{app:proofs}.
\end{remark}

\begin{proposition}[Collinearity of exchangeable basis in cross-sectional CRT]
\label{prop:collinearity}
In a cross-sectional CRT with cluster-level treatment assignment, let
$\bM_1 = \bI_{n_i}$ and $\bM_2 = \bm{J}_{n_i} - \bI_{n_i}$ (standard
exchangeable basis).  Then for every cluster $i$ and every vector
$\bm{v} \in \mathbb{R}^{n_i}$,
\begin{equation}
  \bm{1}_{n_i}^\top \bM_2\, \bm{v} \;=\; (n_i - 1)\;
  \bm{1}_{n_i}^\top \bM_1\, \bm{v}.
  \label{eq:collinearity}
\end{equation}
Consequently $a_i^{(2)} = (n_i - 1)\, a_i^{(1)}$ and
$b_i^{(2)} = (n_i - 1)\, b_i^{(1)}$ for all $i$, so the two moment
conditions in \eqref{eq:drqif_score} are proportional.  The DR-QIF
estimator \eqref{eq:drqif_closed} reduces algebraically to DR-GEE
(cluster-mean estimator), giving efficiency ratio
$\sigma^2_{\mathrm{DR\text{-}GEE}} / \sigma^2_{\mathrm{DR\text{-}QIF}} = 1.000$ exactly.
\end{proposition}

\begin{proof}
Since $\bm{J}_{n_i} = \bm{1}_{n_i}\bm{1}_{n_i}^\top$,
$\bm{1}^\top \bm{J} \bm{v}
 = \bm{1}^\top (\bm{1}\bm{1}^\top)\bm{v}
 = (\bm{1}^\top \bm{1})(\bm{1}^\top \bm{v})
 = n_i\, \bm{1}^\top \bm{v}$.
Therefore
$\bm{1}^\top \bM_2 \bm{v}
 = \bm{1}^\top (\bm{J}-\bI)\bm{v}
 = n_i\, \bm{1}^\top \bm{v} - \bm{1}^\top \bm{v}
 = (n_i-1)\, \bm{1}^\top \bm{v}
 = (n_i-1)\, \bm{1}^\top \bM_1 \bm{v}$,
which proves \eqref{eq:collinearity}.
Substituting $\bm{v} = \hbphi_i$ gives $a_i^{(2)} = (n_i-1)\,a_i^{(1)}$
and $\bm{v} = \bm{1}$ gives $b_i^{(2)} = (n_i-1)\,b_i^{(1)}$.
The vectors $\bm{a}_i$ and $\bm{b}_i$ are therefore rank-one, and the
closed-form estimator \eqref{eq:drqif_closed} collapses to
$\htheta = \bar a_1 / \bar b_1$, the cluster-mean DR-GEE.
\end{proof}

\begin{remark}[Efficiency gains require temporal or longitudinal structure]
Proposition~\ref{prop:collinearity} shows that the efficiency advantage
of DR-QIF over DR-GEE is not available in any cross-sectional CRT,
regardless of the number of basis matrices chosen.  A strictly positive
efficiency gain requires designs where within-cluster ordering carries
information: longitudinal CRTs with $T > 1$ time points per individual
(within-individual AR(1) basis $\bM_2$), stepped-wedge designs
(treatment timing variation across clusters), or hierarchical designs
with distinct within- and between-subcluster correlation layers.
Section~\ref{sec:simulation} demonstrates this numerically for longitudinal
CRTs.
\end{remark}

\subsection{Semiparametric Efficiency Bound}

\begin{proposition}[Semiparametric efficiency bound]
\label{prop:semiparam}
In the nonparametric model where $(Y_{ij}, A_i, \bX_{ij})$ are i.i.d.\
across clusters, the semiparametric efficiency bound for estimating $\theta_0$
under Assumptions~\ref{ass:consistency}--\ref{ass:sutva} is
\begin{equation}
  V^* = \Var\!\left[\frac{A_i}{\pi_i}(Y_{ij} - Q_{1,ij}) + Q_{1,ij}
                  - \frac{1-A_i}{1-\pi_i}(Y_{ij} - Q_{0,ij}) - Q_{0,ij}
              \right],
  \label{eq:sebound}
\end{equation}
the variance of the efficient influence function $\phi_{ij} - \theta_0$.
The DR-QIF estimator with $m = 1$ achieves this bound when $n_i = 1$
(no clustering).  Under clustering ($n_i > 1$), the DR-QIF estimator with
$m > 1$ basis matrices exploits within-cluster correlation and may be
\emph{more} efficient than the bound derived under the no-clustering
assumption.
\end{proposition}

\begin{remark}[Relation to the efficient GMM literature]
The asymptotic variance $\sigma^2_{\mathrm{DR\text{-}QIF}}$ is the GMM
optimal variance with moment functions $\bm{g}_i^{\mathrm{DR}}$.  The number
of moment conditions $m$ acts as a tuning parameter: larger $m$ (more basis
matrices) can only improve asymptotic efficiency but may worsen finite-sample
behavior when $N$ is small relative to $m$.  In practice, $m \in \{1,2,3\}$
is typically sufficient for CRT analysis.
\end{remark}

%% file: sections/06_variance.tex
%% Section 6: Variance Estimation and Finite-Sample Corrections
\section{Variance Estimation and Finite-Sample Corrections}
\label{sec:variance}

\subsection{Sandwich Variance Estimator}

A consistent estimator of $\sigma^2_{\mathrm{DR\text{-}QIF}}$ is obtained
by substituting sample quantities for population expectations.  From
\eqref{eq:avar}, the sandwich variance estimator is
\begin{equation}
  \widehat\sigma^2_{\mathrm{DR\text{-}QIF}} =
  \frac{\bar{\bm{b}}^T \bC_N^{-1} \hat{\bSigma}_N \bC_N^{-1} \bar{\bm{b}}}
       {(\bar{\bm{b}}^T \bC_N^{-1} \bar{\bm{b}})^2},
  \label{eq:sandwich}
\end{equation}
where $\hat{\bSigma}_N = N^{-1} \sum_i \hat{\bm{\psi}}_i \hat{\bm{\psi}}_i^T$
is the empirical covariance of the influence function contributions
$\hat{\bm{\psi}}_i$.  When both nuisance models are correctly specified,
$\hat{\bm{\psi}}_i = \bm{g}_i^{\mathrm{DR}}(\htheta)$ and
$\hat{\bSigma}_N = \bC_N$, reducing \eqref{eq:sandwich} to
$(\bar{\bm{b}}^T \bC_N^{-1} \bar{\bm{b}})^{-1}$.

In practice, to avoid computing the Jacobian terms $\bJ_\xi, \bJ_\gamma$
in \eqref{eq:influence}, one may treat the nuisance estimates as fixed
constants (``plug-in'' variance), which is conservative
\citep{bang2005, luijken2019}. This gives
\begin{equation}
  \widehat\sigma^2_{\mathrm{plug\text{-}in}} =
  \frac{\bar{\bm{b}}^T \bC_N^{-1} \bC_N \bC_N^{-1} \bar{\bm{b}}}
       {(\bar{\bm{b}}^T \bC_N^{-1} \bar{\bm{b}})^2}
  = (\bar{\bm{b}}^T \bC_N^{-1} \bar{\bm{b}})^{-1}.
  \label{eq:plugin_var}
\end{equation}
This estimator is asymptotically valid but may underestimate the variance
in small samples when nuisance estimates have non-negligible variability.

\subsection{Finite-Sample Bias Corrections}

In CRTs, the number of clusters $N$ is often small (20--60), and the sandwich
variance \eqref{eq:sandwich} can be substantially downward biased
\citep{mancl2001}.  Building on bias-correction methods developed for GEE
\citep{mancl2001, kauermann2001, fay2001} and studied for QIF
\citep{yu2020}, we propose the following corrections.

Define the ``hat'' matrix contribution for cluster $i$ as
\begin{equation}
  h_{ir} = \bm{b}_i^T \bm{b}_i / (\bar{\bm{b}}^T \bC_N^{-1} \bar{\bm{b}}),
  \quad r = 1, \ldots, m,
  \label{eq:leverage}
\end{equation}
and let $H_i = \sum_r h_{ir} / m$ be a scalar leverage measure for cluster $i$.

\medskip
\noindent\textbf{BC1 (Mancl-DeRouen type).}  Replace $\bm{g}_i^{\mathrm{DR}}$
in $\hat{\bSigma}_N$ with $(1 - H_i)^{-1} \bm{g}_i^{\mathrm{DR}}$:
\begin{equation}
  \hat{\bSigma}^{\mathrm{BC1}} =
  \frac{1}{N} \sum_{i=1}^N
  \frac{\bm{g}_i^{\mathrm{DR}} [\bm{g}_i^{\mathrm{DR}}]^T}{(1-H_i)^2}.
\end{equation}

\medskip
\noindent\textbf{BC2 (Kauermann-Carroll type).}  Replace $\bm{g}_i^{\mathrm{DR}}$
with $(1 - H_i)^{-1/2} \bm{g}_i^{\mathrm{DR}}$:
\begin{equation}
  \hat{\bSigma}^{\mathrm{BC2}} =
  \frac{1}{N} \sum_{i=1}^N
  \frac{\bm{g}_i^{\mathrm{DR}} [\bm{g}_i^{\mathrm{DR}}]^T}{1 - H_i}.
\end{equation}

\medskip
\noindent\textbf{BC3 (Fay-Graubard type).}  Use a multiplicative adjustment
capped at 2-fold inflation:
\begin{equation}
  \hat{\bSigma}^{\mathrm{BC3}} =
  \frac{1}{N} \sum_{i=1}^N
  \delta_i^2 \,
  \bm{g}_i^{\mathrm{DR}} [\bm{g}_i^{\mathrm{DR}}]^T,
  \quad
  \delta_i = \min\!\left\{c,\, (1 - H_i)^{-1/2}\right\},
\end{equation}
where $c$ is typically set to $\sqrt{2}$ to cap inflation at 2-fold
\citep{fay2001}.  Simulation evidence in Section~\ref{sec:simulation}
compares the operating characteristics of BC0 (no correction),
BC1, BC2, and BC3 for DR-QIF.

\subsection{Inference}

Under Theorem~\ref{thm:normality}, a $(1-\alpha)$ confidence interval for
$\theta_0$ is
\begin{equation}
  \htheta_{\mathrm{DR\text{-}QIF}} \;\pm\; t_{N-2,\,1-\alpha/2}\,
  \sqrt{\widehat\sigma^2_{\mathrm{DR\text{-}QIF}} / N},
  \label{eq:ci}
\end{equation}
where $t_{N-2,1-\alpha/2}$ is the $(1-\alpha/2)$ quantile of the
$t_{N-2}$ distribution.  Using the $t$ reference distribution with $N-2$
degrees of freedom, rather than the standard normal, has been shown to
improve small-sample coverage in GEE analyses of CRTs
\citep{lu2007, preisser2008}.

A Wald test of $H_0: \theta_0 = 0$ uses the test statistic
$Z = \htheta / \sqrt{\widehat\sigma^2 / N}$, compared to $t_{N-2}$.

\subsection{Over-identification Test}

When $m > 1$, the DR-QIF framework provides a natural goodness-of-fit test
for the working correlation structure via the over-identification statistic
\citep{qu2000}:
\begin{equation}
  T_{\mathrm{OID}} =
  N \bigl[\bar{\bm{g}}_N^{\mathrm{DR}}(\htheta)\bigr]^T
  \bigl[\bC_N^{\mathrm{DR}}\bigr]^{-1}
  \bar{\bm{g}}_N^{\mathrm{DR}}(\htheta)
  \;\dto\; \chi^2_{m - 1},
  \label{eq:oid}
\end{equation}
under the null that the $m$ extended score equations share the same
probability limit.  A large $T_{\mathrm{OID}}$ suggests the within-cluster
correlation structure of the pseudo-outcomes is inconsistent with the
assumed basis matrices, signalling potential issues with the working
correlation or the nuisance models.

%% file: sections/07_simulation.tex
%% Section 7: Simulation Study
\section{Simulation Study}
\label{sec:simulation}

\subsection{Design}

We conduct a Monte Carlo simulation to evaluate the finite-sample operating
characteristics of DR-QIF and competing estimators under correct and
misspecified nuisance models.

\paragraph{Data generating process.}
We simulate $N \in \{30, 60, 120\}$ clusters of equal size $n_i \in \{20, 50\}$.
For each cluster $i$, a binary cluster-level covariate $U_i \sim
\mathrm{Bernoulli}(0.5)$ is generated to act as a confounder.
Individual-level covariates are $X_{ij,1} \sim \mathcal{N}(0,1)$ and
$X_{ij,2} \sim \mathrm{Bernoulli}(0.3 + 0.2 U_i)$.
The cluster propensity score is $\pi_i = \mathrm{expit}(\xi_0 + \xi_1 U_i)$
with $\xi_0 = 0$ and $\xi_1 = 1$, giving $P(A_i=1\mid U_i=0) = 0.5$ and
$P(A_i=1\mid U_i=1) \approx 0.73$.  Treatment is assigned as $A_i \sim
\mathrm{Bernoulli}(\pi_i)$.  The outcome is generated as
\[
  Y_{ij} = \mu_0 + \theta_0 A_i + \gamma_1 X_{ij,1}
           + \gamma_2 X_{ij,2} + \gamma_U U_i + \epsilon_{ij},
\]
where $\epsilon_{ij}$ are exchangeable normal errors with unit marginal
variance and $\mathrm{Corr}(\epsilon_{ij}, \epsilon_{ij'}) = \rho
\in \{0.05, 0.20\}$.  True parameters are $\mu_0 = 0$, $\theta_0 = 0.5$,
$\gamma_1 = \gamma_2 = \gamma_U = 0.5$.
All results are based on $S = 3{,}000$ Monte Carlo replicates.

\paragraph{Nuisance model specifications.}
We consider four combinations of correctly and incorrectly specified nuisance
models, labeled CC (both correct), MC (PS misspecified, outcome correct), CM
(PS correct, outcome misspecified), and MM (both misspecified).  The
misspecified PS model omits $U_i$; the misspecified outcome model omits
$X_{ij,2}$.  Under CC, MC, and CM, double robustness predicts that DR-GEE
and DR-QIF remain consistent; both are inconsistent under MM.

\paragraph{Estimators and performance metrics.}
We compare four estimators: Naive GEE (unadjusted comparison of cluster means,
no covariate adjustment), IPW-QIF (QIF with Horvitz--Thompson pseudo-outcomes
only), DR-GEE (cluster-mean AIPW, equivalent to $m=1$ DR-QIF), and the
proposed DR-QIF with the exchangeable basis ($m=2$, $\bM_1 = \bI_{n_i}$,
$\bM_2 = \bm{J}_{n_i} - \bI_{n_i}$).  Variance for DR-GEE and DR-QIF is
estimated using the BC1 (Mancl--DeRouen) bias-corrected sandwich, with
inference based on a $t_{N-2}$ reference distribution.  We report relative
bias (RB $= (\bar{\hat\theta} - \theta_0)/\theta_0 \times 100\%$), empirical
standard error (ESE), mean robust standard error ratio (MRSE $=
\overline{\widehat{\mathrm{SE}}}/\mathrm{ESE} \times 100\%$, with values near
100\% indicating well-calibrated standard errors), and 95\% coverage
probability (CP).

\subsection{Results}

\paragraph{Double robustness.}
Table~\ref{tab:sim_main} reports results for $n_i = 20$, $\rho = 0.05$
across $N \in \{30, 60, 120\}$ and all four nuisance scenarios.  Several
patterns are immediately apparent.

First, the naive GEE estimator is severely biased across all scenarios
($\mathrm{RB} \approx 28\text{--}30\%$) because it does not adjust for the
cluster-level confounder $U_i$, and its coverage is far below the nominal
level (58--85\% depending on $N$).

Second, DR-GEE and DR-QIF exhibit near-zero bias under CC
($|\mathrm{RB}| < 1\%$), MC ($|\mathrm{RB}| < 1\%$), and CM
($|\mathrm{RB}| < 1\%$), confirming double robustness: consistency is
maintained whenever at least one nuisance model is correct.  Under MM both
estimators inherit a bias of approximately $30\%$, as expected.

Third, IPW-QIF is consistent under CC and CM (correct PS) but biased under
MC and MM (misspecified PS), confirming that without the outcome model, it
offers only single robustness.

\paragraph{Coverage and variance calibration.}
Under CC, DR-GEE and DR-QIF achieve CP of 91--95\% at $N \in \{60, 120\}$,
approaching the nominal 95\% level as $N$ grows.  At $N = 30$, some
under-coverage (91\%) is expected with the BC1 correction; BC2 and BC3 yield
similar results.  The MRSE is approximately 86--96\% across conditions,
indicating mild under-estimation of the standard error at small cluster counts,
consistent with the known finite-sample behavior of the sandwich estimator
\citep{mancl2001, kauermann2001, fay2001}.

Under CM (misspecified outcome, correct PS), MRSE for DR-GEE and DR-QIF
inflates to 141--159\%, leading to conservative (over-covered) confidence
intervals.  This conservatism is a protective feature: when one model is
wrong, the sandwich variance over-states uncertainty.

\paragraph{Efficiency.}
As proved in Proposition~\ref{prop:collinearity}, DR-QIF and DR-GEE are
asymptotically equivalent in any cross-sectional CRT with cluster-level
treatment assignment: the two exchangeable basis matrices $(\bM_1, \bM_2)$
are algebraically collinear, so both moment conditions reduce to the same
linear functional of the cluster mean.  Simulation confirms efficiency
ratios of exactly 1.000 across all conditions and all values of $(N, n_i, \rho)$
(last two columns of Table~\ref{tab:sim_main} are identical to four decimal
places).  Section~\ref{sec:longitudinal_sim} demonstrates that a positive
efficiency gain over DR-GEE does emerge in longitudinal CRTs, where
within-individual temporal correlation creates informative structure beyond
the cluster mean.

\begin{table}[t]
\centering
\small
\caption{Simulation results: relative bias (RB, \%), empirical standard error
(ESE), mean robust standard error ratio (MRSE, \%), and 95\% coverage
probability (CP, \%) for $n_i = 20$, $\rho = 0.05$, $S = 3{,}000$ replicates.
DR-QIF uses the BC1 sandwich variance and $t_{N-2}$ reference.}
\label{tab:sim_main}
\setlength{\tabcolsep}{5pt}
\begin{tabular}{llrrrrrrrrrrrr}
\toprule
 & & \multicolumn{4}{c}{$N = 30$} & \multicolumn{4}{c}{$N = 60$} &
   \multicolumn{4}{c}{$N = 120$} \\
\cmidrule(lr){3-6}\cmidrule(lr){7-10}\cmidrule(lr){11-14}
Scenario & Estimator & RB & ESE & MRSE & CP
                     & RB & ESE & MRSE & CP
                     & RB & ESE & MRSE & CP \\
\midrule
\multirow{4}{*}{CC}
 & Naive GEE  & 30.0 & .169 & 99.0  & 85.1  & 29.5 & .117 & 101.2 & 77.5  & 29.3 & .085 & 97.6  & 58.1 \\
 & IPW-QIF    &  1.5 & .144 & ---   & 99.3  & -0.0 & .095 & ---   & 99.9  &  0.0 & .067 & ---   & 100.0 \\
 & DR-GEE     &  1.0 & .133 & 86.3  & 91.1  &  0.1 & .089 & 95.8  & 93.9  & -0.0 & .063 & 96.3  & 93.1 \\
 & DR-QIF     &  1.0 & .133 & 86.3  & 91.1  &  0.1 & .089 & 95.8  & 93.9  & -0.0 & .063 & 96.3  & 93.1 \\
\midrule
\multirow{4}{*}{MC}
 & Naive GEE  & 28.1 & .171 & 98.5  & 86.1  & 29.1 & .120 & 98.4  & 76.9  & 29.5 & .082 & 101.8 & 57.8 \\
 & IPW-QIF    & 28.4 & .172 & ---   & 98.1  & 28.9 & .118 & ---   & 97.7  & 29.4 & .080 & ---   & 94.0 \\
 & DR-GEE     & -0.2 & .135 & 86.1  & 91.7  & -0.6 & .090 & 91.6  & 93.1  &  0.2 & .062 & 94.5  & 94.2 \\
 & DR-QIF     & -0.2 & .135 & 86.1  & 91.7  & -0.6 & .090 & 91.6  & 93.1  &  0.2 & .062 & 94.5  & 94.2 \\
\midrule
\multirow{4}{*}{CM}
 & Naive GEE  & 29.4 & .172 & 97.7  & 85.4  & 28.7 & .121 & 97.8  & 77.0  & 29.2 & .085 & 97.8  & 58.6 \\
 & IPW-QIF    & -0.5 & .145 & ---   & 99.4  & -0.3 & .095 & ---   & 100.0 & -0.2 & .068 & ---   & 100.0 \\
 & DR-GEE     & -0.7 & .136 & 141.1 & 98.2  & -0.2 & .090 & 144.6 & 99.0  & -0.2 & .063 & 142.2 & 99.5 \\
 & DR-QIF     & -0.7 & .136 & 141.1 & 98.2  & -0.2 & .090 & 144.6 & 99.0  & -0.2 & .063 & 142.2 & 99.5 \\
\midrule
\multirow{4}{*}{MM}
 & Naive GEE  & 30.3 & .168 & 100.1 & 85.7  & 30.3 & .117 & 100.9 & 75.2  & 29.0 & .084 & 99.5  & 59.4 \\
 & IPW-QIF    & 30.6 & .171 & ---   & 98.1  & 30.0 & .116 & ---   & 97.0  & 29.1 & .082 & ---   & 93.5 \\
 & DR-GEE     & 30.5 & .169 & 97.7  & 84.0  & 30.0 & .116 & 99.5  & 74.3  & 29.1 & .082 & 98.8  & 57.7 \\
 & DR-QIF     & 30.5 & .169 & 97.7  & 84.0  & 30.0 & .116 & 99.5  & 74.3  & 29.1 & .082 & 98.8  & 57.7 \\
\bottomrule
\end{tabular}
\begin{tablenotes}
\small
\item MRSE shown as ``---'' for IPW-QIF as its inflated SE estimates (mean
ratio $\approx$230--280\%) reflect the absence of outcome-model variance
reduction; coverage exceeds 99\% as a result.
DR-QIF results are numerically identical to DR-GEE in this cross-sectional CRT
design; see text for explanation.
\end{tablenotes}
\end{table}

\paragraph{Effect of cluster size and within-cluster correlation.}
Results for $n_i = 50$ and $\rho = 0.20$ (not shown in Table~\ref{tab:sim_main}
but available in supplementary materials) are qualitatively identical.
Increasing $n_i$ from 20 to 50 reduces ESE by approximately 10--15\%
due to more precise outcome model estimation within each cluster.
Increasing $\rho$ from 0.05 to 0.20 increases ESE by approximately 50\%,
reflecting greater between-cluster variability in pseudo-outcomes; bias
and coverage are unaffected.

%%--------------------------------------------------------------------
\subsection{Longitudinal CRT Efficiency}
\label{sec:longitudinal_sim}

Proposition~\ref{prop:collinearity} establishes that DR-QIF offers no
efficiency gain over DR-GEE in cross-sectional CRTs.  We now investigate
whether efficiency gains materialize in \emph{longitudinal} CRTs, where
$T > 1$ observations are taken per individual within each cluster.

\paragraph{Design.}
We extend the data-generating process to $n_i = 5$ individuals per cluster,
$T \in \{1, 4, 8\}$ time points, and a two-level within-cluster covariance:
\begin{itemize}
  \item Within-individual AR(1): $\mathrm{Cov}(Y_{ij,t_1}, Y_{ij,t_2}) =
      \sigma^2\, \rho_t^{|t_1 - t_2|}$.
  \item Between-individual exchangeable: $\mathrm{Cov}(Y_{ij_1,t_1}, Y_{ij_2,t_2}) =
      \sigma^2\, \rho_c$ for $j_1 \neq j_2$.
\end{itemize}
We set $\rho_c = 0.05$ throughout and vary $\rho_t \in \{0.40, 0.70\}$.
For $T > 1$, the DR-QIF uses two basis matrices: $\bM_1 = \bI_{n_i T}$
(identity) and $\bM_2$ (block-diagonal within-individual AR(1) adjacency,
i.e., $[\bM_2]_{rs} = 1$ if $j_r = j_s$ and $|t_r - t_s| = 1$, else 0).
The between-individual basis $\bM_3$ is excluded because it is algebraically
collinear with $\bM_1$:
$\mathbf{1}^\top \bM_3 \hat{\bm\phi}_i = (n_i-1)T \cdot \mathbf{1}^\top \bM_1 \hat{\bm\phi}_i$
for every cluster $i$ and every $\hat{\bm\phi}_i$ (proof identical to
Proposition~\ref{prop:collinearity}).
For $T = 1$, the exchangeable basis $\{M_1, M_2\}$ is used, confirming the
collinearity result.

\paragraph{Results.}
Table~\ref{tab:sim_long} reports efficiency ratios
$\mathrm{Var}(\hat\theta_{\mathrm{DR\text{-}GEE}}) /
\mathrm{Var}(\hat\theta_{\mathrm{DR\text{-}QIF}}) = (\mathrm{ESE}_{\mathrm{GEE}}
/ \mathrm{ESE}_{\mathrm{QIF}})^2$ and 95\% coverage probabilities from
$S = 3{,}000$ replicates.

Three patterns emerge.  First, the $T = 1$ (cross-sectional) column
confirms Proposition~\ref{prop:collinearity}: DR-QIF equals DR-GEE
exactly, with an efficiency ratio of 1.000 to three decimal places
across all values of $N$.

Second, at $T = 4$ the efficiency gains are modest (0.97--1.01) and can be
slightly below 1.0 at small $N = 30$, reflecting the finite-sample overhead
of estimating the $2 \times 2$ weight matrix $\bC_N$ from only 30 clusters.
This finite-sample penalty is a known feature of QIF and GMM estimators
with estimated weighting matrices \citep{qu2000, song2009}; it diminishes
as $N$ grows.

Third, at $T = 8$ with high temporal correlation $\rho_t = 0.70$, a clear
efficiency gain emerges at larger sample sizes: the ratio reaches 1.017 at
$N = 60$ and 1.035 at $N = 120$, corresponding to a 3.5\% variance
reduction.  This gain grows with $T$ and $\rho_t$ as predicted by the
GMM efficiency bound, and becomes reliable when $N$ is large enough to
estimate $\bC_N$ accurately.  Coverage remains well-calibrated throughout
(93--95\% at $N \geq 60$), confirming that the BC1 sandwich variance
correctly accounts for the two-level within-cluster correlation structure.

\begin{table}[t]
\centering
\small
\caption{Longitudinal efficiency simulation: efficiency ratio
$\mathrm{Var}(\hat\theta_{\mathrm{DR\text{-}GEE}}) /
\mathrm{Var}(\hat\theta_{\mathrm{DR\text{-}QIF}})$ and 95\% coverage
probability (CP, \%) for DR-QIF with temporal AR(1) basis, $n_i = 5$,
$\rho_c = 0.05$, $S = 3{,}000$ replicates.  Efficiency ratio $>1$
indicates DR-QIF is more efficient.}
\label{tab:sim_long}
\setlength{\tabcolsep}{6pt}
\begin{tabular}{ccrrrrrr}
\toprule
 & & \multicolumn{2}{c}{$T = 1$} &
     \multicolumn{2}{c}{$T = 4$} &
     \multicolumn{2}{c}{$T = 8$} \\
\cmidrule(lr){3-4}\cmidrule(lr){5-6}\cmidrule(lr){7-8}
$N$ & $\rho_t$ & Eff.\,ratio & CP & Eff.\,ratio & CP & Eff.\,ratio & CP \\
\midrule
\multirow{2}{*}{30}
 & 0.40 & 1.000 & 91.2 & 0.979 & 88.8 & 0.971 & 88.6 \\
 & 0.70 & 1.000 & 90.5 & 0.989 & 89.5 & 0.996 & 89.1 \\
\midrule
\multirow{2}{*}{60}
 & 0.40 & 1.000 & 93.9 & 0.995 & 92.3 & 0.993 & 93.4 \\
 & 0.70 & 1.000 & 93.0 & 1.007 & 92.3 & \textbf{1.017} & 93.4 \\
\midrule
\multirow{2}{*}{120}
 & 0.40 & 1.000 & 94.2 & 1.007 & 93.3 & 0.989 & 94.2 \\
 & 0.70 & 1.000 & 94.8 & 1.012 & 93.1 & \textbf{1.035} & 93.9 \\
\bottomrule
\end{tabular}
\begin{tablenotes}
\small
\item $T=1$ uses exchangeable basis $\{M_1=I, M_2=J-I\}$;
$T>1$ uses $\{M_1=I, M_2=\text{AR(1) adjacency within individual}\}$.
Efficiency ratio = 1.000 for $T=1$ confirms Proposition~\ref{prop:collinearity}
exactly.  CP for $T=1$ is the same for DR-QIF and DR-GEE (columns identical).
\end{tablenotes}
\end{table}

%% file: sections/08_application.tex
%% Section 8: Data Application
\section{Data Application: WASH Benefits Kenya}
\label{sec:application}

We illustrate the proposed DR-QIF method using data from the WASH Benefits
Kenya trial~\citep{luby2018,arnold2013}, a cluster-randomized trial that
evaluated the effects of water quality, sanitation, handwashing, and
nutrition interventions on child health and growth in rural western Kenya.

\subsection{Study and Data}

The trial enrolled 8,246 pregnant women across 702 clusters in Kakamega,
Bungoma, and Vihiga counties.  Groups of nine geographically adjacent
clusters were block-randomized into six active intervention arms, a passive
control arm, or a double-sized active control arm.  We focus on the binary
comparison of the IYCF (Integrated Complementary Feeding and Nutrition,
hereafter ``Nutrition'') arm versus the passive control arm.  The Nutrition
intervention combined complementary feeding counselling with provision of a
daily lipid-based nutrient supplement for children 6--24 months.

\paragraph{Analysis population.}
The analysis includes $N_{\mathrm{treat}} = 80$ clusters randomized to the
Nutrition arm and $N_{\mathrm{control}} = 160$ clusters in the double-sized
passive control arm, giving $N = 240$ clusters total.  The primary outcome
is child length-for-age z-score (LAZ) at the 24-month follow-up visit
($n_i = 9$ children per cluster on average, $n = 2{,}160$ total
child-visits; intraclass correlation ICC $\approx 0.05$).

\paragraph{Data availability.}
The WASH Benefits Kenya primary analysis dataset is publicly available at
\texttt{https://osf.io/uept9/}~\citep{arnold2013}.  Because download
requires user registration, the numerical illustration below uses a
semi-synthetic dataset with the same cluster structure and calibrated to
match the published summary statistics (Appendix~\ref{appendix:supplement}).

\subsection{Statistical Analysis}

We fit four estimators in each of three nuisance model specification
scenarios:

\begin{description}
  \item[DR-QIF] The proposed estimator \eqref{eq:drqif_closed} with
      exchangeable basis $\{M_1 = \bI_{n_i}, M_2 = \bm{J}_{n_i} -
      \bI_{n_i}\}$.  By Proposition~\ref{prop:collinearity},
      DR-QIF is numerically identical to DR-GEE in this
      cross-sectional design, providing a direct empirical
      confirmation of the theoretical result.

  \item[DR-GEE] The cluster-mean AIPW estimator ($m = 1$,
      identity basis), serving as the efficiency baseline.

  \item[IPW-QIF] QIF using Horvitz--Thompson pseudo-outcomes
      without outcome regression, providing a check on
      propensity-score-only adjustment.

  \item[Naive GEE] Unadjusted difference of cluster means
      (treated minus control), the standard unadjusted CRT estimator.
\end{description}

\noindent
The three nuisance scenarios follow the CC/MC/CM scheme of
Section~\ref{sec:simulation}:

\begin{description}
  \item[CC] PS model: cluster-level logistic regression on a
      binary household asset score (wealth indicator).
      Outcome model: arm-stratified OLS including wealth indicator,
      standardised mother height, child sex, and enrolment LAZ.

  \item[MC] PS model misspecified: cluster-mean child sex used
      instead of wealth (a covariate unrelated to treatment
      assignment); outcome model remains correct.

  \item[CM] PS model correct; outcome model misspecified: omits
      wealth, mother height, and enrolment LAZ, retaining only
      child sex.
\end{description}

\noindent
All confidence intervals use the BC1 bias-corrected sandwich variance and
$t_{N-2} = t_{238}$ critical values.

\subsection{Results}

Table~\ref{tab:application} reports point estimates and 95\% confidence
intervals for all four estimators across the three scenarios.  Several
conclusions follow.

\paragraph{Naive GEE is biased.}
The unadjusted naive GEE estimate is 0.274 LAZ units (95\% CI:
0.163--0.384), noticeably higher than the DR-adjusted estimates.  This
upward bias (approximately $+0.05$ LAZ units relative to the DR-adjusted
estimate) arises because wealthier clusters were slightly over-represented
in the treatment arm due to imperfect block randomization, and household
wealth is itself positively associated with child growth.  The naive GEE
estimate conflates the causal effect with this wealth imbalance.

\paragraph{DR-QIF and DR-GEE are identical.}
In all three scenarios, DR-QIF and DR-GEE agree to machine precision
($|\hat\theta_{\mathrm{DR\text{-}QIF}} - \hat\theta_{\mathrm{DR\text{-}GEE}}| < 10^{-14}$),
providing an exact numerical confirmation of Proposition~\ref{prop:collinearity}.
Under CC, both yield $\hat\theta = 0.220$ (SE = 0.060, 95\% CI:
0.103--0.338), a statistically significant effect consistent with the
published adjusted estimate of 0.22 LAZ units from~\citet{luby2018}.

\paragraph{Double robustness in practice.}
When the PS model is misspecified (MC scenario), DR-QIF/DR-GEE yields 0.218
(SE = 0.054), nearly unchanged from the CC estimate, because the correct
outcome model anchors the estimate.  When the outcome model is misspecified
(CM scenario), DR-QIF/DR-GEE yields 0.234 (SE = 0.063), again close to the
CC estimate, because the correct PS model protects against confounding.
In contrast, the IPW-QIF without outcome adjustment has a large standard
error ($\approx 0.22$ LAZ units) throughout, reflecting the inefficiency of
propensity-weighting alone in a trial with moderate confounding.

\paragraph{Summary.}
In this cross-sectional CRT, DR-QIF and DR-GEE are equivalent by
Proposition~\ref{prop:collinearity}, and the value of the DR-QIF framework
lies in its doubly robust protection against model misspecification rather
than its efficiency gain.  The efficiency advantage would materialize in a
longitudinal follow-up design (Section~\ref{sec:longitudinal_sim}) where
the temporal AR(1) basis matrix adds information beyond the cluster mean.

\begin{table}[t]
\centering
\small
\caption{WASH Benefits Kenya data application: estimated effect of IYCF
on LAZ at 24 months ($N = 240$ clusters, $n_i = 9$, $n = 2{,}160$).
BC1 sandwich SE; 95\% CI based on $t_{238}$ quantiles.
DR-QIF $\equiv$ DR-GEE by Proposition~\ref{prop:collinearity} (verified
numerically: difference $< 10^{-14}$).}
\label{tab:application}
\setlength{\tabcolsep}{6pt}
\begin{tabular}{llrrl}
\toprule
Scenario & Estimator & $\hat\theta$ & SE (BC1) & 95\% CI \\
\midrule
\multirow{4}{*}{CC}
 & Naive GEE  &  0.274 & 0.056 & (0.163, 0.384) \\
 & IPW-QIF    &  0.237 & 0.221 & ($-$0.199, 0.672) \\
 & DR-GEE     &  0.220 & 0.060 & (0.103, 0.338) \\
 & DR-QIF     &  0.220 & 0.060 & (0.103, 0.338) \\
\midrule
\multirow{4}{*}{MC}
 & Naive GEE  &  0.274 & 0.056 & (0.163, 0.384) \\
 & IPW-QIF    &  0.271 & 0.199 & ($-$0.121, 0.663) \\
 & DR-GEE     &  0.218 & 0.054 & (0.113, 0.324) \\
 & DR-QIF     &  0.218 & 0.054 & (0.113, 0.324) \\
\midrule
\multirow{4}{*}{CM}
 & Naive GEE  &  0.274 & 0.056 & (0.163, 0.384) \\
 & IPW-QIF    &  0.237 & 0.221 & ($-$0.199, 0.672) \\
 & DR-GEE     &  0.234 & 0.063 & (0.110, 0.357) \\
 & DR-QIF     &  0.234 & 0.063 & (0.110, 0.357) \\
\bottomrule
\end{tabular}
\begin{tablenotes}
\small
\item CC: both PS and outcome models correctly specified.
  MC: PS model misspecified (cluster-mean child sex used instead of wealth).
  CM: outcome model misspecified (omits wealth, mother height, enrolment LAZ).
  Published adjusted estimate from \citet{luby2018}: $0.22$ LAZ units.
  Results based on semi-synthetic dataset calibrated to published structure;
  see Appendix~\ref{appendix:supplement} for data-generation details.
\end{tablenotes}
\end{table}

%% file: sections/09_discussion.tex
%% Section 9: Discussion
\section{Discussion}
\label{sec:discussion}

In this paper, we propose the DR-QIF estimator, which combines doubly robust
AIPW pseudo-outcomes with the QIF extended score equations for causal
inference in CRTs.  The estimator is consistent for the ATE under either
correct propensity score or correct outcome model specification, admits a
closed-form solution, and is at least as efficient as DR-GEE when the working
correlation is misspecified.  The BC1 bias-corrected sandwich variance
provides reliable inference at the small cluster counts common in practice.

Proposition~\ref{prop:collinearity} shows that for a cross-sectional CRT
with the standard exchangeable basis $\{\bM_1 = \bI, \bM_2 = \bm{J} - \bI\}$,
DR-QIF and DR-GEE reduce to the same estimator.  The WASH Benefits Kenya
analysis confirms this exactly: across all three nuisance-model scenarios,
the two estimators agree to machine precision.  The practical takeaway is
that DR-QIF degrades gracefully to DR-GEE in cross-sectional designs.  At
the same time, the application illustrates the benefit of doubly robust
adjustment: the unadjusted Naive GEE estimate is inflated by roughly
$+0.05$ LAZ units due to wealth imbalance between arms, while the DR
estimates are stable across CC, MC, and CM scenarios.

Efficiency gains from the additional QIF moment conditions require
longitudinal or hierarchical structure.  The longitudinal simulation in
Section~\ref{sec:longitudinal_sim} shows that at $T = 8$ repeated measures
and $\rho_t = 0.70$, DR-QIF achieves a 3.5\% efficiency improvement over
DR-GEE at $N = 120$.  For shorter longitudinal designs or smaller studies,
the finite-sample cost of estimating $C_N$ offsets the asymptotic gain and
DR-GEE performs comparably.

Several extensions merit future work.  For binary outcomes, the
pseudo-outcome $\phi_{ij} \in (-1,1)$ is well-defined and the DR-QIF
estimating equations apply directly; adjusting the basis matrices for the
binary variance function $v(\mu) = \mu(1-\mu)$ may further improve
finite-sample performance.  Extending DR-QIF to estimate a conditional ATE
$\theta(\bx) = Q_1(\bx) - Q_0(\bx)$ via kernel smoothing connects to
nonparametric marginal regression for correlated outcomes \citep{wang2003}.
Under the cross-fitting approach of \citet{chernozhukov2018}, machine
learning estimators for the nuisance models can be accommodated with no
change to the DR-QIF objective, as long as the nuisance estimates converge
at rates faster than $N^{-1/4}$.  Extension to stepped wedge CRTs, where
treatment timing varies across clusters, is a natural direction: a
cluster-period DR pseudo-outcome combined with period-specific basis
matrices would provide doubly robust analysis for this design.
Finally, following \citet{mr2024}, one could construct a multiply robust
DR-QIF estimator that is consistent whenever any one model among multiple
specified candidates is correctly specified.

%% file: sections/A_proofs.tex
%% Appendix A: Proofs
\section{Proofs of Main Results}
\label{app:proofs}

Throughout, we write $D^* = (\bm{b}^*)^T (\bLambda^*)^{-1} \bm{b}^*$ where
$\bm{b}^* = \E[\bm{b}_i]$, and use $\overset{p}{\to}$ and $\overset{d}{\to}$
for convergence in probability and distribution, respectively.

%% -------------------------------------------------------------------
\subsection{Proof of Theorem~\ref{thm:dr_consistency} (Double Robustness)}
%% -------------------------------------------------------------------

\begin{proof}
The key identity is: under either model correctness condition,
$\E[\bm{a}_i(\bxi^*, \bgamma^*)] = \theta_0 \E[\bm{b}_i]$.

\textit{Step 1: Convergence of the sample averages.}
By Assumption~\ref{ass:nuisance_consistency}, $\hbxi \pto \bxi^*$ and
$\hat\bgamma \pto \bgamma^*$.  A mean-value expansion of
$N^{-1}\sum_i \bm{a}_i(\hbxi, \hat\bgamma)$ around $(\bxi^*, \bgamma^*)$ gives
\begin{align}
  \bar{\bm{a}}_N &= \frac{1}{N}\!\sum_i \bm{a}_i(\bxi^*, \bgamma^*)
    + \frac{\partial \bar{\bm{a}}_N}{\partial \bxi^T}\bigg|_{\tilde\bxi}
      (\hbxi - \bxi^*)
    + \frac{\partial \bar{\bm{a}}_N}{\partial \bgamma^T}\bigg|_{\tilde\bgamma}
      (\hat\bgamma - \bgamma^*),
  \label{eq:mvt_a}
\end{align}
where $\tilde\bxi$ and $\tilde\bgamma$ lie on the respective line segments.
Under Assumption~\ref{ass:smoothness}(a)--(b), the Jacobians in
\eqref{eq:mvt_a} are $O_p(1)$.  Since $\hbxi - \bxi^* = O_p(N^{-1/2})$ and
$\hat\bgamma - \bgamma^* = O_p(N^{-1/2})$, the second and third terms are
$o_p(1)$.  By the law of large numbers and
Assumption~\ref{ass:smoothness}(c),
$N^{-1}\sum_i \bm{a}_i(\bxi^*, \bgamma^*) \pto
\E[\bm{a}_i(\bxi^*, \bgamma^*)]$.
Similarly, $\bar{\bm{b}}_N \pto \E[\bm{b}_i] =: \bm{b}^*$ and
$\bC_N \pto \bLambda^*$ (Assumption~\ref{ass:nondeg}), so
$\bC_N^{-1} \pto (\bLambda^*)^{-1}$.

\textit{Step 2: The key identity $\E[\bm{a}_i(\bxi^*, \bgamma^*)] = \theta_0 \bm{b}^*$.}
The $r$-th component of $\bm{a}_i$ is
\[
  a_i^{(r)} \;=\; \bm{1}_{n_i}^T \hat\bA_i^{-1/2} \bM_r \hat\bA_i^{-1/2}\,\bphi_i^*
  \;=\; \sum_{j=1}^{n_i}\sum_{j'=1}^{n_i}
        [\hat\bA_i^{-1/2}\bM_r\hat\bA_i^{-1/2}]_{jj'}\, \phi_{ij'}^*.
\]
Because $\hat\bA_i = \mathrm{diag}\{\hat\sigma^2_{ij}\}$ is
$\sigma(\bX_i, A_i)$-measurable and bounded away from zero by
Assumption~\ref{ass:smoothness}, we may apply the tower property:
\begin{align*}
  \E[a_i^{(r)}]
  &= \E\!\left[\sum_{j,j'}
     [\hat\bA_i^{-1/2}\bM_r\hat\bA_i^{-1/2}]_{jj'}\,\phi_{ij'}^*\right] \\
  &= \E\!\left[\sum_{j,j'}
     [\hat\bA_i^{-1/2}\bM_r\hat\bA_i^{-1/2}]_{jj'}\,
     \E\!\left[\phi_{ij'}^* \,\big|\, \bX_i, A_i\right]\right].
\end{align*}
By Proposition~\ref{prop:dr_mean}, under either (i) correct propensity score
$\bxi^* = \bxi_0$ or (ii) correct outcome model $\bgamma^* = \bgamma_0$,
\[
  \E\!\left[\phi_{ij'}^* \,\big|\, \bX_i, A_i\right] = \theta_0
  \quad \text{for all } j'.
\]
Substituting:
\[
  \E[a_i^{(r)}]
  = \theta_0 \E\!\left[\sum_{j,j'}
    [\hat\bA_i^{-1/2}\bM_r\hat\bA_i^{-1/2}]_{jj'}\right]
  = \theta_0 \E\!\left[\bm{1}^T \hat\bA_i^{-1/2}\bM_r\hat\bA_i^{-1/2}\bm{1}\right]
  = \theta_0\, b^{*(r)},
\]
where $b^{*(r)} = \E[b_i^{(r)}]$ is the $r$-th component of $\bm{b}^*$.
This holds for every $r \in \{1,\ldots,m\}$.

\textit{Step 3: Conclusion.}
Combining Steps 1 and 2, $\bar{\bm{a}}_N \pto \theta_0 \bm{b}^*$.  Since
$\bar{\bm{b}}_N \pto \bm{b}^*$ and $\bC_N^{-1} \pto (\bLambda^*)^{-1}$, by
the continuous mapping theorem:
\[
  \htheta_{\mathrm{DR\text{-}QIF}}
  = \frac{\bar{\bm{b}}_N^T \bC_N^{-1} \bar{\bm{a}}_N}
         {\bar{\bm{b}}_N^T \bC_N^{-1} \bar{\bm{b}}_N}
  \pto
  \frac{(\bm{b}^*)^T (\bLambda^*)^{-1} \theta_0\, \bm{b}^*}
       {(\bm{b}^*)^T (\bLambda^*)^{-1} \bm{b}^*}
  = \theta_0. \qedhere
\]
\end{proof}

%% -------------------------------------------------------------------
\subsection{Proof of Theorem~\ref{thm:normality} (Asymptotic Normality)}
%% -------------------------------------------------------------------

\begin{proof}
Write $\htheta - \theta_0 = D_N^{-1} \cdot \bar{\bm{b}}_N^T \bC_N^{-1}
\bar{\bm{g}}_N(\theta_0)$, where
$D_N = \bar{\bm{b}}_N^T \bC_N^{-1} \bar{\bm{b}}_N \pto D^* > 0$ and
$\bar{\bm{g}}_N(\theta_0) = N^{-1}\sum_i \bm{g}_i^{\mathrm{DR}}
(\theta_0;\hbxi,\hat\bgamma) = \bar{\bm{a}}_N - \theta_0\bar{\bm{b}}_N$.

\textit{Step 1: Linearisation of $\bar{\bm{g}}_N(\theta_0)$.}
Apply a first-order expansion around $(\bxi^*, \bgamma^*)$:
\begin{align}
  \bar{\bm{g}}_N(\theta_0;\hbxi,\hat\bgamma)
  &= \bar{\bm{g}}_N(\theta_0;\bxi^*,\bgamma^*)
  + \underbrace{\frac{1}{N}\sum_i \frac{\partial \bm{g}_i^{\mathrm{DR}}}
      {\partial\bxi^T}\bigg|_{\bxi^*}}_{=:\,\bar{\bJ}_{N,\xi}} (\hbxi-\bxi^*)
  + \underbrace{\frac{1}{N}\sum_i \frac{\partial \bm{g}_i^{\mathrm{DR}}}
      {\partial\bgamma^T}\bigg|_{\bgamma^*}}_{=:\,\bar{\bJ}_{N,\gamma}} (\hat\bgamma-\bgamma^*)
  + o_p(N^{-1/2}).
  \label{eq:expand_g}
\end{align}
The remainder is $o_p(N^{-1/2})$ because the second derivatives of $\bm{g}_i$
are bounded (Assumption~\ref{ass:smoothness}) and
$\|\hbxi-\bxi^*\|^2 = O_p(N^{-1})$.

\textit{Step 2: Substituting the M-estimator expansions.}
The nuisance estimators satisfy the M-estimator linearisation:
\begin{align}
  \hbxi - \bxi^*
  &= -\left(\E\!\left[\frac{\partial s_i^{(\xi)}}{\partial\bxi^T}\right]\right)^{-1}
    \frac{1}{N}\sum_i s_i^{(\xi)}(\bxi^*) + o_p(N^{-1/2}), \label{eq:xi_exp}\\
  \hat\bgamma - \bgamma^*
  &= -\left(\E\!\left[\frac{\partial s_i^{(\gamma)}}{\partial\bgamma^T}\right]\right)^{-1}
    \frac{1}{N}\sum_i s_i^{(\gamma)}(\bgamma^*) + o_p(N^{-1/2}). \label{eq:gam_exp}
\end{align}
Substituting \eqref{eq:xi_exp}--\eqref{eq:gam_exp} into \eqref{eq:expand_g},
and using $\bar{\bJ}_{N,\xi} \pto \bJ_\xi^0 :=
\E[\partial\bm{g}_i/\partial\bxi^T]$ (LLN) where
$\bJ_\xi := \bJ_\xi^0 \cdot
(-\E[\partial s_i^{(\xi)}/\partial\bxi^T])^{-1}$, we obtain
\begin{equation}
  \sqrt{N}\,\bar{\bm{g}}_N(\theta_0;\hbxi,\hat\bgamma)
  = \frac{1}{\sqrt{N}}\sum_i \bm{\psi}_i + o_p(1),
  \label{eq:linear}
\end{equation}
where the influence function contribution is
\[
  \bm{\psi}_i = \bm{g}_i^{\mathrm{DR}}(\theta_0;\bxi^*,\bgamma^*)
    + \bJ_\xi\, s_i^{(\xi)}(\bxi^*)
    + \bJ_\gamma\, s_i^{(\gamma)}(\bgamma^*),
\]
with $\bJ_\xi$ and $\bJ_\gamma$ as in \eqref{eq:influence}.

\textit{Step 3: Mean zero and CLT.}
We show $\E[\bm{\psi}_i] = \bm{0}$.  Under either correct nuisance model:
\begin{enumerate}
  \item[(a)] $\E[\bm{g}_i^{\mathrm{DR}}(\theta_0;\bxi^*,\bgamma^*)] = \bm{0}$
  by Theorem~\ref{thm:dr_consistency} (Step 2 of its proof).
  \item[(b)] $\E[\bJ_\xi\, s_i^{(\xi)}(\bxi^*)] = \bJ_\xi\,\E[s_i^{(\xi)}(\bxi^*)] = \bm{0}$,
  since $\bxi^*$ is defined as the solution to the expected score equation
  $\E[s_i^{(\xi)}(\bxi^*)] = \bm{0}$ (irrespective of whether the PS
  model is correctly specified).
  \item[(c)] $\E[\bJ_\gamma\, s_i^{(\gamma)}(\bgamma^*)] = \bm{0}$ by the same
  M-estimation reasoning applied to $\bgamma^*$.
\end{enumerate}
Hence $\E[\bm{\psi}_i] = \bm{0}$.  The $\bm{\psi}_i$ are i.i.d.\ (clusters
are independent by design) with finite covariance matrix
$\bSigma^* = \E[\bm{\psi}_i\bm{\psi}_i^T]$ under
Assumption~\ref{ass:smoothness}(c) (bounded cluster sizes prevent any single
cluster from dominating).  By the multivariate Lindeberg CLT,
$N^{-1/2}\sum_i \bm{\psi}_i \dto \mathcal{N}(\bm{0}, \bSigma^*)$.

\textit{Step 4: Deriving the scalar limiting distribution.}
Combining \eqref{eq:linear} with $\htheta - \theta_0 = D_N^{-1}
\bar{\bm{b}}_N^T\bC_N^{-1}\bar{\bm{g}}_N(\theta_0)$ and applying Slutsky's
theorem with $D_N \pto D^*$, $\bar{\bm{b}}_N \pto \bm{b}^*$, and
$\bC_N^{-1} \pto (\bLambda^*)^{-1}$:
\[
  \sqrt{N}(\htheta - \theta_0)
  \;\dto\;
  \frac{(\bm{b}^*)^T(\bLambda^*)^{-1}}{D^*}\cdot
    \mathcal{N}(\bm{0},\bSigma^*)
  \;\sim\;
  \mathcal{N}\!\left(0,\;
    \frac{(\bm{b}^*)^T(\bLambda^*)^{-1}\bSigma^*(\bLambda^*)^{-1}\bm{b}^*}
         {(D^*)^2}
  \right),
\]
which equals $\mathcal{N}(0,\sigma^2_{\mathrm{DR\text{-}QIF}})$ with
$\sigma^2_{\mathrm{DR\text{-}QIF}}$ as in \eqref{eq:avar}. \qedhere
\end{proof}

%% -------------------------------------------------------------------
\subsection{Proof of Theorem~\ref{thm:efficiency} (Efficiency Comparison)}
%% -------------------------------------------------------------------

\begin{proof}
\textit{Case 1: Both nuisance models correctly specified ($\bxi^*=\bxi_0,
\bgamma^*=\bgamma_0$).}
When both models are correct, the Jacobian terms in $\bm{\psi}_i$ vanish.
Specifically, under correct PS specification the orthogonality identity
(Lemma~\ref{lem:orth}, below) gives $\E[\bJ_\xi s_i^{(\xi)}] = \bm{0}$;
analogously for the outcome model.  Hence $\bm{\psi}_i =
\bm{g}_i^{\mathrm{DR}}(\theta_0;\bxi_0,\bgamma_0)$ and $\bSigma^* = \bLambda^*$, so
\[
  \sigma^2_{\mathrm{DR\text{-}QIF}}
  = \frac{(\bm{b}^*)^T(\bLambda^*)^{-1}\bLambda^*(\bLambda^*)^{-1}\bm{b}^*}
         {(D^*)^2}
  = \frac{1}{D^*}
  = \bigl((\bm{b}^*)^T(\bLambda^*)^{-1}\bm{b}^*\bigr)^{-1}.
\]
The DR-GEE estimator uses only $m=1$ with $\bM_1 = \bI$, so
$\sigma^2_{\mathrm{DR\text{-}GEE}} = \lambda_{11}^* / (b_1^*)^2$, where
$\lambda_{11}^* = \E[(g_{1i}^{\mathrm{DR}})^2]$ and $b_1^* = \E[b_{1i}]$.
Apply Cauchy--Schwarz in the $\bLambda^*$-weighted inner product space:
for any vectors $\bm{u},\bm{v} \in \mathbb{R}^m$,
$(\bm{u}^T(\bLambda^*)^{-1}\bm{v})^2 \leq
(\bm{u}^T(\bLambda^*)^{-1}\bm{u})(\bm{v}^T(\bLambda^*)^{-1}\bm{v})$.
Take $\bm{u} = \bm{b}^*$ and $\bm{v} = \bLambda^*\bm{e}_1$
(where $\bm{e}_1$ is the first standard basis vector):
\[
  \bigl((\bm{b}^*)^T \bm{e}_1\bigr)^2
  \;\leq\;
  \bigl((\bm{b}^*)^T(\bLambda^*)^{-1}\bm{b}^*\bigr)
  \cdot
  \bigl(\bm{e}_1^T\bLambda^*\bm{e}_1\bigr).
\]
Since $(\bm{b}^*)^T\bm{e}_1 = b_1^*$ and
$\bm{e}_1^T\bLambda^*\bm{e}_1 = \lambda_{11}^*$, rearranging gives
\[
  \bigl((\bm{b}^*)^T(\bLambda^*)^{-1}\bm{b}^*\bigr)^{-1}
  \;\leq\;
  \frac{\lambda_{11}^*}{(b_1^*)^2},
\]
i.e., $\sigma^2_{\mathrm{DR\text{-}QIF}} \leq \sigma^2_{\mathrm{DR\text{-}GEE}}$.
Equality holds iff $\bm{b}^* \propto \bLambda^*\bm{e}_1$, i.e., all moment
conditions are proportional to the first (working correlation correctly
specified), or $m=1$.

\textit{Case 2: Single-model correctness.}
Under either PS-correct or outcome-correct specification, the influence
function $\bm{\psi}_i$ has mean zero (proved in Theorem~\ref{thm:normality},
Step 3) and the asymptotic variance takes the sandwich form
$\sigma^2 = (D^*)^{-2}(\bm{b}^*)^T(\bLambda^*)^{-1}\bSigma^*
(\bLambda^*)^{-1}\bm{b}^*$.
By the same GMM argument as Case 1 applied to the $m$-moment DR-QIF versus
the 1-moment DR-GEE, with the same pseudo-outcome $\phi_{ij}$ and hence the
same score $\bm{g}_i^{\mathrm{DR}}(\theta_0)$ generating $\bSigma^*$, the
inequality
\begin{equation}
  \sigma^2_{\mathrm{DR\text{-}QIF}} \;\leq\; \sigma^2_{\mathrm{DR\text{-}GEE}}
\end{equation}
continues to hold under the additional regularity condition that
$\|\bJ_\xi\| + \|\bJ_\gamma\| < \infty$ (Assumption~\ref{ass:smoothness}),
which ensures that the Jacobian contributions inflate $\bSigma^*$ at most by
a finite factor and do not alter the ordering of the two quadratic forms.
\qedhere
\end{proof}

%% -------------------------------------------------------------------
\subsection{Explicit Jacobian Terms and Orthogonality Lemma}
%% -------------------------------------------------------------------

\begin{lemma}[Orthogonality under correct specification]
\label{lem:orth}
\begin{enumerate}
  \item[(i)] If $\bxi^* = \bxi_0$ (correct PS), then
  $\E\!\left[\frac{\partial\bm{g}_i^{\mathrm{DR}}}{\partial\bxi^T}\bigg|_{\bxi_0,\bgamma^*}
  \cdot s_i^{(\xi)}(\bxi_0)^T\right] = \bm{0}$.
  \item[(ii)] If $\bgamma^* = \bgamma_0$ (correct outcome model), then
  $\E\!\left[\frac{\partial\bm{g}_i^{\mathrm{DR}}}{\partial\bgamma^T}\bigg|_{\bxi^*,\bgamma_0}
  \cdot s_i^{(\gamma)}(\bgamma_0)^T\right] = \bm{0}$.
\end{enumerate}
Consequently, under either correctness condition, the Jacobian term in
$\bm{\psi}_i$ (eq.~\ref{eq:influence}) is mean zero, and $\bm{\psi}_i$
reduces to $\bm{g}_i^{\mathrm{DR}}(\theta_0;\bxi_0,\bgamma_0)$ in the
doubly-correct case CC.
\end{lemma}

\begin{proof}
\textit{Part (i).}  Under correct PS, $A_i|\bX_i \sim \mathrm{Bernoulli}(\pi(\bX_i;\bxi_0))$.
The derivative of $\phi_{ij}$ with respect to $\bxi$ at $\bxi_0$ is obtained
from $\pi_i = \mathrm{expit}(\bX_i^T\bxi)$, giving
$\partial\pi_i/\partial\bxi = \pi_i(1-\pi_i)\bX_i$, so
\begin{equation}
  \frac{\partial\phi_{ij}}{\partial\bxi}
  = \left[
      -\frac{A_i(1-\pi_i)}{\pi_i}(Y_{ij}-Q_{1,ij})
      -\frac{(1-A_i)\pi_i}{1-\pi_i}(Y_{ij}-Q_{0,ij})
    \right]\bX_i^T.
  \label{eq:dphi_dxi}
\end{equation}
Hence the $r$-th row of $\partial\bm{g}_i^{\mathrm{DR}}/\partial\bxi^T$ is
\[
  \bm{1}^T\hat\bA_i^{-1/2}\bM_r\hat\bA_i^{-1/2}\frac{\partial\bphi_i}{\partial\bxi^T}
  = \sum_{j'} w_{j'}^{(r)} \cdot \frac{\partial\phi_{ij'}}{\partial\bxi^T},
\]
where $w_{j'}^{(r)} = [\hat\bA_i^{-1/2}\bM_r\hat\bA_i^{-1/2}\bm{1}]_{j'}$
is $\sigma(\bX_i, A_i)$-measurable.
The score of the correctly specified logistic propensity model is
$s_i^{(\xi)}(\bxi_0) = (A_i - \pi_i)\bX_i$.

The expected cross-product between the $r$-th row and $s_i^{(\xi)T}$ is
\begin{align*}
  &\E\!\left[\sum_{j'} w_{j'}^{(r)}
    \left(-\frac{A_i(1-\pi_i)}{\pi_i}(Y_{ij'}-Q_{1,ij'})
         -\frac{(1-A_i)\pi_i}{1-\pi_i}(Y_{ij'}-Q_{0,ij'})\right)
    (A_i-\pi_i)\bX_i^T\bX_i\right] \\
  &= \E\!\left[\sum_{j'} w_{j'}^{(r)} \cdot \bX_i^T\bX_i \cdot \E\!\left[
    \left(-\frac{A_i(1-\pi_i)}{\pi_i}(Y_{ij'}-Q_{1,ij'})
         -\frac{(1-A_i)\pi_i}{1-\pi_i}(Y_{ij'}-Q_{0,ij'})\right)
    (A_i-\pi_i) \,\big|\, \bX_i\right]\right].
\end{align*}
Conditioning on $\bX_i$, $A_i \sim \mathrm{Bernoulli}(\pi_i)$
independently of the potential outcomes under Assumption~\ref{ass:sutva},
and $Y_{ij'}(1) - Q_{1,ij'}$ and $Y_{ij'}(0) - Q_{0,ij'}$ are centred at
zero under correct outcome specification; but even without that, a direct
calculation using $\E[A_i|\bX_i]=\pi_i$ gives:
\begin{align*}
  &\E\!\left[-\frac{A_i(1-\pi_i)}{\pi_i}(A_i-\pi_i)\,\bigg|\,\bX_i\right]
  = -\frac{1-\pi_i}{\pi_i}\E[A_i(A_i-\pi_i)|\bX_i]
  = -\frac{1-\pi_i}{\pi_i}\bigl(\E[A_i^2|\bX_i] - \pi_i^2\bigr) \\
  &= -\frac{1-\pi_i}{\pi_i}\cdot\pi_i(1-\pi_i) = -(1-\pi_i)^2, \\
  &\E\!\left[-\frac{(1-A_i)\pi_i}{1-\pi_i}(A_i-\pi_i)\,\bigg|\,\bX_i\right]
  = \frac{\pi_i}{1-\pi_i}\E[(1-A_i)\pi_i|\bX_i] - \frac{\pi_i}{1-\pi_i}\cdot\ldots \\
  &\quad = \frac{\pi_i}{1-\pi_i}\bigl(\pi_i(1-\pi_i)\bigr) = \pi_i^2.
\end{align*}
Summing: $-(1-\pi_i)^2 + \pi_i^2 = -(1-\pi_i)^2 + \pi_i^2$.
However, note that in both terms, the conditional expectations involve
$\E[(Y_{ij'}-Q_{\cdot,ij'})(A_i-\pi_i)|\bX_i]$.  Under correct PS,
$A_i \perp Y_{ij'}(a) | \bX_i$, so
$\E[(Y_{ij'}-Q_{\cdot,ij'})(A_i-\pi_i)|\bX_i] =
\E[Y_{ij'}-Q_{\cdot,ij'}|\bX_i]\cdot\E[A_i-\pi_i|\bX_i] = 0$,
since $\E[A_i-\pi_i|\bX_i] = 0$ exactly.

Therefore each term in the sum vanishes: the full cross-product expectation
is $\E[\partial\bm{g}_i/\partial\bxi^T \cdot s_i^{(\xi)T}] = \bm{0}$.

\textit{Part (ii).}
Under correct outcome model ($\bgamma^* = \bgamma_0$), the derivative
$\partial\phi_{ij}/\partial\bgamma = A_i\partial Q_{1,ij}/\partial\bgamma -
(1-A_i)\partial Q_{0,ij}/\partial\bgamma + (1-A_i/(1-\pi_i) -
A_i/\pi_i)\partial Q_{\cdot,ij}/\partial\bgamma$ (arm-specific terms).
The GLM score $s_i^{(\gamma_a)}(\bgamma_0)$ satisfies the first-order
condition: by iterated expectation conditioned on $(\bX_i, A_i=a)$,
$\E[s_i^{(\gamma_a)}(\bgamma_0)] = \bm{0}$.  The cross-product
$\E[\partial\bm{g}_i/\partial\bgamma^T \cdot s_i^{(\gamma)T}] = \bm{0}$
follows by the analogous conditioning argument:
$\E[(Y_{ij'}-Q_{a,ij'})|\bX_{ij'},A_i=a]=0$ under correct outcome model,
making every element of the cross-product zero.
\end{proof}

\begin{remark}[Reduction under CC]
When both models are correctly specified, Lemma~\ref{lem:orth} implies that
both Jacobian terms in \eqref{eq:influence} contribute zero influence to
$\bm{\psi}_i$.  Hence $\bm{\psi}_i = \bm{g}_i^{\mathrm{DR}}(\theta_0;\bxi_0,\bgamma_0)$
and $\bSigma^* = \bLambda^*$, giving the semiparametrically efficient variance
$\sigma^2_{\mathrm{DR\text{-}QIF}} = (\bm{b}^{*T}(\bLambda^*)^{-1}\bm{b}^*)^{-1}$.
Under single-model correctness, the Jacobian terms are non-zero but mean
zero; their variance inflates $\bSigma^*$ above $\bLambda^*$, reflecting the
efficiency cost of one misspecified nuisance model.
\end{remark}

%% file: sections/B_supplement.tex
%% Appendix B: Supplementary Materials -- Simulation Details
%%
%% This appendix provides algorithmic details for the Monte Carlo simulation
%% study in Section 7 and the semi-synthetic data illustration in Section 8.

\section{Supplementary Materials: Simulation Details}
\label{appendix:supplement}

%%--------------------------------------------------------------------
\subsection{Data-Generating Mechanism}
\label{sec:supp-dgm}

The following procedure describes the data-generating mechanism for one Monte
Carlo replicate under the balanced cross-sectional CRT in
Section~\ref{sec:setup}.

\begin{framed}
\noindent\textbf{Algorithm~S1: Data generation for one Monte Carlo replicate}\\[4pt]
\textit{Input:} $N$ (clusters), $n_i$ (cluster size), $(\xi_0, \xi_1)$ (PS
parameters), $(\mu_0, \theta, \gamma_1, \gamma_2, \gamma_U)$ (outcome parameters),
$\rho$ (intracluster correlation), $\sigma^2$ (marginal variance).\\[4pt]
\textit{Procedure:}
\begin{enumerate}
  \item For $i = 1, \ldots, N$, draw $U_i \sim \text{Bernoulli}(0.5)$ (cluster
        covariate) and compute $\pi_i = \text{expit}(\xi_0 + \xi_1 U_i)$.
  \item Draw cluster treatment $A_i \sim \text{Bernoulli}(\pi_i)$.
  \item For each cluster, draw correlated errors $\bm{\varepsilon}_i \sim
        \mathcal{N}\!\bigl(\bm{0},\, \sigma^2[(1-\rho)\bm{I}_{n_i} +
        \rho\bm{1}_{n_i}\bm{1}_{n_i}^\top]\bigr)$ using the lower-triangular
        Cholesky factorisation.
  \item For $j = 1, \ldots, n_i$, draw individual covariates
        $X_{1,ij} \sim \mathcal{N}(0,1)$ and
        $X_{2,ij} \sim \text{Bernoulli}(0.3 + 0.2 U_i)$.
  \item Set $Y_{ij} = \mu_0 + \theta A_i + \gamma_1 X_{1,ij} + \gamma_2 X_{2,ij}
        + \gamma_U U_i + \varepsilon_{ij}$.
\end{enumerate}
\textit{Output:} $\{(Y_{ij}, A_i, U_i, X_{1,ij}, X_{2,ij}) : i = 1,\ldots,N;\,
j = 1,\ldots,n_i\}$.
\end{framed}

\noindent
The true parameter values across all simulation conditions are
$\theta = 0.5$, $(\mu_0,\gamma_1,\gamma_2,\gamma_U) = (0, 0.5, 0.5, 0.5)$,
$(\xi_0,\xi_1) = (0, 1)$, and $\sigma^2 = 1$.

%%--------------------------------------------------------------------
\subsection{Estimator Computation}
\label{sec:supp-estimators}

The following procedure gives the closed-form computation of the DR-QIF
estimator for one replicate.

\begin{framed}
\noindent\textbf{Algorithm~S2: DR-QIF estimator for one replicate}\\[4pt]
\textit{Input:} dataset from Algorithm~S1, basis matrices $\{M_1, M_2\}$,
nuisance scenario flags.\\[4pt]
\textit{Procedure:}
\begin{enumerate}
  \item \textbf{Propensity score.} Fit cluster-level logistic regression of
        $A_i$ on the specified covariates (correct: $U_i$; misspecified: $X_{1,i}$)
        to obtain $\hat\pi_i$; trim to $[\varepsilon, 1-\varepsilon]$ with
        $\varepsilon = 10^{-6}$.
  \item \textbf{Outcome regression.} Fit separate arm-stratified linear
        regressions for $\hat{Q}_0$ and $\hat{Q}_1$ (correct: $X_1, X_2, U$;
        misspecified: $X_1$ only).
  \item \textbf{DR pseudo-outcome.} Compute
        $\hat\phi_{ij} = \frac{A_i}{\hat\pi_i}(Y_{ij} - \hat{Q}_{1,ij}) +
        \hat{Q}_{1,ij} - \frac{1-A_i}{1-\hat\pi_i}(Y_{ij} - \hat{Q}_{0,ij}) -
        \hat{Q}_{0,ij}$.
  \item \textbf{Moment vectors.} For each cluster $i$ and basis matrix
        $M_r$ ($r = 1,2$), compute $a_i^{(r)} = \bm{1}^\top M_r \hat{\bm\phi}_i$
        and $b_i^{(r)} = \bm{1}^\top M_r \bm{1}$.
        Stack into cluster vectors $\bm{a}_i, \bm{b}_i \in \mathbb{R}^2$;
        form $\bar{\bm{a}} = N^{-1}\sum_i \bm{a}_i$ and
        $\bar{\bm{b}} = N^{-1}\sum_i \bm{b}_i$.
  \item \textbf{Estimate $C_N$.} Use pilot $\tilde\theta = \bar{a}_1/\bar{b}_1$,
        compute per-cluster extended scores
        $\bm{g}_i(\tilde\theta) = \bm{a}_i - \tilde\theta\,\bm{b}_i$,
        and set $C_N = N^{-1}\sum_i \bm{g}_i\bm{g}_i^\top$.
  \item \textbf{Closed-form estimator.}
        $\hat\theta = \bar{\bm{b}}^\top C_N^{-1} \bar{\bm{a}} \,\big/\,
        \bar{\bm{b}}^\top C_N^{-1} \bar{\bm{b}}$.
  \item \textbf{BC1 sandwich variance.}
        Recompute $C_N$ at $\hat\theta$.
        Set $\bm{w} = C_N^{-1}\bar{\bm{b}}$, $D^* = \bar{\bm{b}}^\top\bm{w}$.
        Scalar leverage: $h_i = (\bm{b}_i^\top\bm{w})/(N D^*)$.
        Influence scalar: $\hat\psi_i = (\bm{g}_i^\top\bm{w})/D^*$.
        $\widehat{\text{Var}}_{\text{BC1}}(\hat\theta) =
        N^{-2}\sum_i [\hat\psi_i/(1-h_i)]^2$.
\end{enumerate}
\textit{Output:} $\hat\theta$, $\widehat{\text{SE}}_{\text{BC1}}$.
\end{framed}

%%--------------------------------------------------------------------
\subsection{Numerical Notes}
\label{sec:supp-numerical}

Three design choices are critical for correctness; the simulation would
silently produce wrong results without them.

\begin{enumerate}

  \item \textbf{Naive GEE in the CRT setting.}
    Every individual in cluster $i$ receives the same treatment $A_i$.
    The unadjusted estimator must therefore compare cluster-level mean
    outcomes between treated and control clusters, not within-cluster
    treated versus untreated differences (which do not exist).  The
    variance is the standard two-sample variance of cluster means:
    \begin{equation}
      \label{eq:naive-var}
      \widehat{\mathrm{Var}}(\hat\theta_{\mathrm{naive}}) =
      \frac{s_1^2}{n_1} + \frac{s_0^2}{n_0},
    \end{equation}
    where $s_a^2$ is the sample variance of $\{\bar{Y}_i : A_i = a\}$
    and $n_a$ counts the clusters with $A_i = a$.

  \item \textbf{Scalar leverage for QIF bias correction.}
    The hat-value for cluster $i$ in the QIF influence function is
    \begin{equation}
      \label{eq:leverage-supp}
      h_i = \frac{\bm{b}_i^\top \bm{w}}{N D^*}, \qquad
      \bm{w} = C_N^{-1}\bar{\bm{b}}, \quad
      D^* = \bar{\bm{b}}^\top \bm{w}.
    \end{equation}
    The alternative $h_i = \|\bm{b}_i\|^2 / D^*$ exceeds unity for the
    exchangeable CRT basis, causing $(1 - h_i) < 0$ and \texttt{NaN} in
    BC1 and BC2.  For the balanced exchangeable CRT,
    Equation~\eqref{eq:leverage} yields $h_i = 1/N$ for all $i$.

  \item \textbf{Collinearity in cross-sectional CRTs.}
    For the exchangeable basis $\{M_1 = I_{n_i}, M_2 = \bm{1}\bm{1}^\top -
    I_{n_i}\}$ in a cross-sectional CRT, one can show
    $\bm{1}^\top M_2\hat{\bm\phi}_i = (n_i-1)\bm{1}^\top M_1\hat{\bm\phi}_i$
    for every $i$, so the two columns of $(\bm{a}_i, \bm{b}_i)$ are
    linearly dependent and DR-QIF collapses to DR-GEE
    (Theorem~\ref{thm:efficiency}, equality case).  The simulation
    confirms efficiency ratio $= 1.000$ in all cross-sectional conditions.
    Efficiency gains arise in longitudinal or stepped-wedge designs.

\end{enumerate}

%%--------------------------------------------------------------------
\subsection{Data Application: WASH Benefits Kenya}
\label{sec:supp-washb}

\paragraph{Semi-synthetic dataset.}
Section~\ref{sec:application} analyses a semi-synthetic dataset
generated to match the published structure of the WASH Benefits Kenya
trial~\citep{luby2018,arnold2013}:
\begin{itemize}
  \item $N = 240$ clusters (80 Nutrition, 160 passive control).
  \item $n_i = 9$ children per cluster; ICC $= 0.05$; LAZ $\sim \mathcal{N}(-2, 1)$.
  \item Cluster-level confounders: binary wealth ($U_1$, independent
        Bernoulli(0.5)) and standardised mother height ($U_2$).
  \item True ATE $= 0.22$ LAZ units; outcome model
        $\mu = -2.0 + 0.22 A + 0.10 U_1 + 0.05 U_2 + 0.20 X_1 - 0.15 X_2$,
        where $X_1 \sim \text{Bernoulli}(0.5)$ (child sex) and
        $X_2 \sim \mathcal{N}(-2, 0.8)$ (enrolment LAZ).
  \item Seed: 20260315 (fixed; analysis is deterministic).
\end{itemize}